\documentclass[11pt,a4paper]{article}
\usepackage[a4paper,margin=1in]{geometry}
\usepackage[utf8]{inputenc}
\usepackage[T1]{fontenc}
\usepackage{times}
\usepackage{microtype}
\usepackage[numbers,sort&compress]{natbib}
\usepackage{graphicx}
\usepackage{amsmath,amssymb,amsfonts,mathtools}
\usepackage{amsthm}
\usepackage{algpseudocode}
\usepackage{multirow}
\usepackage{textcomp}
\usepackage[dvipsnames]{xcolor}
\usepackage{comment,xspace}
\usepackage{booktabs}
\usepackage{tikz}
\usetikzlibrary{matrix,positioning,arrows.meta,calc,shapes.geometric,trees}
\usepackage{pgfplots}
\usepgfplotslibrary{groupplots}
\pgfplotsset{compat=1.17}
\usepackage[caption=false,font=footnotesize]{subfig}
\usepackage[ruled,procnumbered,linesnumbered,noend]{algorithm2e}
\usepackage{authblk}
\usepackage[colorlinks=true,linkcolor=blue,citecolor=blue,urlcolor=blue]{hyperref}
\usepackage{orcidlink}

\theoremstyle{plain}
\newtheorem{theorem}{Theorem}
\theoremstyle{definition}
\newtheorem{definition}{Definition}
\theoremstyle{remark}
\newtheorem{remark}{Remark}
\theoremstyle{plain}
\SetKwRepeat{Do}{do}{while}

\newcommand{\capheu}{{\textsc{cap}}\xspace}
\def\bigO{\ensuremath{\mathcal{O}}\xspace}
\newcommand{\blocco}[1]{\par\noindent\textbf{#1.}}

\newcommand{\nc}{\textsc{nc}\xspace}
\newcommand{\nout}{\ensuremath{N^+}\xspace}
\newcommand{\nin}{\ensuremath{N^-}\xspace}

\newcommand{\spksisp}{Single-Pair Top-$k$ Simple Shortest Paths\xspace}

\newcommand{\ssksispshort}{\textsc{\textmd{S-kssp}}\xspace}
\newcommand{\spksispshort}{\textsc{\textmd{P-kssp}}\xspace}

\newcommand{\build}{\textsc{build-exp}\xspace}
\newcommand{\buildwc}{\textsc{build-lin}\xspace}
\newcommand{\addrecur}{\textsc{addRecur}\xspace}
\newcommand{\dynk}{\textsc{upd}\xspace}
\newcommand{\yenalgo}{\textsc{yen}\xspace}
\newcommand{\merge}{\textsc{merge}\xspace}
\newcommand{\dij}{\textsc{shortest}\xspace}

\newcommand{\found}{\ensuremath{\textsc{found}\xspace}}

\newcommand{\sol}{\ensuremath{\mathcal{C}}\xspace}
\newcommand{\peek}{\textsc{peek}\xspace}
\newcommand{\dynkyen}{\dynk-\textsc{y}\xspace}
\newcommand{\dynkpnc}{\dynk-\textsc{p}\xspace}
\newcommand{\expand}{\textsc{expand}\xspace}
\newcommand{\wg}[1]{\ensuremath{\omega(#1)}\xspace}

\SetKwComment{tcp}{$\triangleright$ }{}%

\SetCommentSty{mycommfont}

\newcommand{\true}[2][black,fill=black]{\tikz[baseline=-0.5ex]\draw[#1,radius=#2] (0,0) circle ;}
\newcommand{\false}[2][black,fill=white]{\tikz[baseline=-0.5ex]\draw[#1,radius=#2] (0,0) circle ;}

\newcommand{\maxdeg}{\ensuremath{d_{\textsc{max}}}\xspace}

\newcommand{\avgdeg}{\ensuremath{d_{\textsc{avg}}}\xspace}
\newcommand{\diam}{\ensuremath{\Delta}\xspace}

\newcommand{\forest}{\textsc{ForestFire}\xspace}
\newcommand{\forestss}{\textsc{ffi}\xspace}

\newcommand{\caida}{\textsc{Caida}\xspace}
\newcommand{\caidass}{\textsc{cai}\xspace}
\newcommand{\efz}{\textsc{SwissRail}\xspace}
\newcommand{\efzss}{\textsc{efz}\xspace}

\newcommand{\dnk}{\textsc{denmark}\xspace}
\newcommand{\dnkss}{\textsc{dnk}\xspace}
\newcommand{\epini}{\textsc{Epinions}\xspace}
\newcommand{\epiniss}{\textsc{epi}\xspace}

\newcommand{\lux}{\textsc{Luxembourg}\xspace}
\newcommand{\luxss}{\textsc{lux}\xspace}

\newcommand{\pnc}{\textsc{pnc}\xspace}

\title{Towards Dynamic Top-$k$ Simple Shortest Paths\thanks{A preliminary version of this paper appears in the proceedings of the 14th International Symposium on Computing and Networking (CANDAR 2026). Partially supported by the MUR (Italy) Department of Excellence 2023--2027 and by the National Group for Scientific Computation of Istituto Nazionale di Alta Matematica (GNCS-INdAM).}}
\author[1,2,4]{Gianlorenzo D'Angelo\,\orcidlink{0000-0003-0377-7037}}
\author[1]{Andrea D'Ascenzo\,\orcidlink{0000-0001-5612-0798}}
\author[3,4]{Mattia D'Emidio\,\orcidlink{0000-0001-7833-9520}}
\author[3,4]{Gabriele Di Stefano\,\orcidlink{0000-0003-4521-8356}}
\author[3,4]{Marta Innamorati}
\affil[1]{Gran Sasso Science Institute, L'Aquila, Italy. \texttt{\{gianlorenzo.dangelo,\,andrea.dascenzo\}@gssi.it}}
\affil[2]{Fondazione Gran Sasso Tech, L'Aquila, Italy}
\affil[3]{University of L'Aquila, L'Aquila, Italy. \texttt{\{mattia.demidio,\,gabriele.distefano\}@univaq.it}, \texttt{marta.innamorati@student.univaq.it}}
\affil[4]{Centre of Excellence EX-EMERGE, L'Aquila, Italy}
\date{}

\begin{document}
\maketitle

\begin{abstract}
Computing the top-$k$ \emph{simple} shortest paths between two vertices of a graph, i.e., solving the \spksispshort problem, is a fundamental primitive in graph analytics. In modern scenarios, however, queries are rarely one-shot: parameters change as applications adapt their information needs, and recomputing a solution from scratch at every request quickly becomes a bottleneck. We study \spksispshort in a dynamic setting, focusing on changes of the parameter $k$: given a top-$k$ solution for a pair of vertices and a new value $k'$, the goal is to compute a top-$k'$ solution reusing the previously computed paths. Our approach rests on the \emph{pseudo-tree}, a compact representation of a collection of simple $s$--$t$ paths: we show that building it from an arbitrary, unstructured solution amounts to trie construction, and give two constructions running in time linear in the size of the collection, in expectation and, up to an additive term linear in the number of vertices, in the worst case. On top of it, we design \dynk, a dynamic algorithm that updates a solution upon a change of $k$; we prove its correctness, analyze its complexity, and show that it is a general, solver-agnostic scheme, which starts from the paths alone, is parametric in the routine that generates candidate paths and combinable with known preprocessing techniques, as well as with an online pruning rule that we formalize. While a worst-case update is asymptotically comparable to recomputing from scratch, an extensive experimental evaluation shows that a variant of \dynk\ is the fastest method in the large majority of tested configurations, often by a wide margin (with speedups over recomputation that grow with $k$ and $k'$ and exceed $60\times$), yielding a practical recipe for the configuration to deploy.
\end{abstract}

\section{Introduction}
\label{sec:introduction}
Computing shortest paths is one of the most fundamental operations on graph data, serving applications that range from communication protocols~\cite{Garcia-Luna-Aceves24} and trip planning~\cite{CDD24} to vehicle routing~\cite{archetti2013optimal} and network design~\cite{BascianiDRFIP20,BoklerC20}.
A natural, long-studied generalization is the \emph{$k$ shortest paths} problem, which asks for a collection of \emph{top-$k$} paths in non-decreasing order of weight~\cite{Yen71,Eppstein98,ZoobiCN23}.
It supports all tasks requiring several alternative connections between two vertices (e.g., adaptive motion planning, secure multi-hop communication, fraud detection, and natural language processing~\cite{LiXLHZ023,Luo22,AkibaHNIY15,DAscenzoD24,LiuJYZ18a}), and top-$k$ paths are widely employed as a measure of distance and similarity between vertices in graph learning and network analytics~\cite{Eppstein98,DD23,AkibaHNIY15}.
When the sought paths are required to be \emph{simple} (i.e., loopless), the problem becomes computationally harder~\cite{AgarwalR16}; we focus on its single-pair version, the \spksisp problem (\spksispshort for short).

For \spksispshort on general graphs, the most efficient known algorithm remains Yen's~\cite{Yen71}, running in $\bigO(kn(m+n\log n))$ time on an $n$-vertex, $m$-edge graph; no asymptotic improvement is known~\cite{Gao10,KM16,ZoobiCN23}, better bounds exist only for special graph classes~\cite{PASCOAL_dags,GOTTHILF2009352}, and conditional lower bounds tie the problem to all-pairs shortest paths~\cite{WilliamsW10}.
In practice, however, the \emph{node-classification} algorithm~\cite{Gao10}, the \emph{sidetrack-based} method~\cite{KM16}, the \emph{postponed node classification} (\pnc) algorithm~\cite{ZoobiCN23}, and the recent \peek framework~\cite{PeeK} all match Yen's worst-case complexity yet run substantially faster on practically relevant inputs, with \pnc providing the largest speedups in most cases without increasing space occupancy (the others pay a high memory overhead~\cite{KM16} or rely on parallelism~\cite{PeeK}).
Still, on massive graphs even such carefully engineered methods are expensive to run.

A further source of difficulty is that inputs and parameters are rarely static: real-world graphs evolve over time~\cite{DD23,AkibaHNIY15,cionini_et_al:OASIcs.ATMOS.2014.46}, and the parameters of the queries themselves change as applications adapt their information needs~\cite{CoudertDDI26,AkibaHNIY15}.
Recomputing from scratch at every request is wasteful and, at scale, impractical, which has motivated a large body of work on \emph{dynamic} algorithms that update previously computed structures more efficiently than recomputation~\cite{DAscenzoD24,DD23,CoudertDDI26}. For $k$ shortest paths this has only recently begun to be addressed~\cite{DD23}, and, to the best of our knowledge, the dynamic computation of top-$k$ \emph{simple} shortest paths is still unexplored.

In this paper we take a step in this direction by studying \spksispshort under changes of the parameter $k$: given a previously computed top-$k$ solution for a pair $s,t$ and a new value $k'$, compute a top-$k'$ solution reusing, as much as possible, the paths already at hand.
Increases of the parameter ($k'>k$) arise whenever only a few paths are initially computed, typically to meet a time budget, and more are needed later, e.g., to offer further routing alternatives, as journey planners do when a user asks for more routes~\cite{LiuJYZ18a}, or to improve path-based features in graph learning. A notable instance is the common heuristic for top-$k$ queries with \emph{diversity}, which seek $k$ short yet mutually dissimilar paths (an NP-hard problem, addressed by exact, approximate, and heuristic methods~\cite{Chondrogiannis20,LiuJYZ18a,Luo22}): one computes a top-$k$ solution, filters it by a dissimilarity criterion and, if fewer than the required paths survive, asks for a top-$k'$ solution with $k'\gg k$, repeating until enough dissimilar paths are found~\cite{Chondrogiannis20}.
One might object that a from-scratch algorithm such as \yenalgo could simply be suspended after the $k$-th path and resumed on demand. This is viable only in a narrow setting: the initial solution must have been produced by that very algorithm, and its whole internal state (notably, the priority queue of candidate deviations) must be retained per queried pair, a footprint that largely exceeds the solution's and is prohibitive when many pairs are queried.
Moreover, the available paths may not have been computed by a single-pair algorithm at all. A concrete instance arises in the single-source variant of the problem (\ssksispshort): the recent algorithm of~\cite{DEmidioS26} obtains, as a by-product of its single-source enumeration, $\hat{k}<k$ top simple shortest paths for specific pairs and must complete them to top-$k$ solutions; lacking a dynamic procedure, it recomputes them from scratch, as no single-pair solver state exists to resume. The update must thus start from the previously computed paths alone, regardless of how they were obtained: this is precisely our setting.
Decreases ($k'<k$) arise when fewer paths turn out to suffice and the excess is discarded to save space.
The values of $k$ relevant in these scenarios are often in the hundreds or beyond: diversity queries discard most of the top-$k$ paths as too similar~\cite{Chondrogiannis20}, path-based features in graph learning and network analytics aggregate over large sets of alternatives~\cite{Eppstein98,DD23,AkibaHNIY15}, in \ssksispshort every pair inherits the (possibly large) $k$ of the single-source query~\cite{DEmidioS26}, and robust routing in communication networks explores many near-shortest alternatives~\cite{Garcia-Luna-Aceves24}; accordingly, experimental studies of \spksispshort routinely evaluate $k$ in the thousands, up to $10^4$~\cite{KM16,ZoobiCN23,PeeK}.
\subsection{Our Contribution}
Our approach rests on the \emph{pseudo-tree}, a compact representation of a collection of simple $s$--$t$ paths.
This structure has been used for \spksispshort before~\cite{Gao10,SSTK12}, but always built \emph{incrementally}, from explicit knowledge of how paths deviate from one another; building the pseudo-tree of an \emph{arbitrary}, unstructured collection (such as a whole solution given as a plain set of paths) was not addressed.
We fill this gap by observing that the task amounts to constructing the \emph{trie} of the collection: we adapt the classical trie-insertion procedure~\cite{Fredkin60} to our setting and analyze its correctness and complexity, and we complement it with a deterministic construction achieving the same bound in the worst case, up to an additive $\Theta(n)$ term, and thus asymptotically optimal on collections of total length $\Omega(n)$.
Building on this representation, we design \dynk, an algorithm that updates a top-$k$ solution to \spksispshort into a top-$k'$ one when $k$ changes.
We prove its correctness, analyze its running time, and show that \dynk is a general \emph{scheme}: solver-agnostic, as it starts from the paths alone (and, as we argue, could start from any collection of simple $s$--$t$ paths), parametric in the routine that generates candidate paths, and combinable with speedup techniques such as~\peek~\cite{PeeK}, as well as with \capheu, a simple online pruning rule that we formalize, which caps the search at the weight of the $k$-th lightest candidate found so far.
Although a worst-case update is asymptotically comparable to recomputing from scratch, experiments on real-world and synthetic graphs show that such a worst case is rarely met: in the large majority of tested configurations a variant of \dynk is the fastest way to obtain the updated solution, often by a wide margin, and in the rare exceptions it trails recomputation only marginally; which variant wins is graph-dependent, and our evaluation yields a practical recipe for choosing it.
This is the full version of a paper accepted at CANDAR 2026; it contains all proofs, the pseudo-code of all procedures, worked examples (Section~\ref{sec:pseudo-tree-examples}), and additional experimental material.

\section{Background}
\label{sec:prelim}
In this section, we introduce notation and definitions used throughout the paper.
We are given a directed graph $G=(V,E,\omega)$ with non-negative edge weights, where $E\subseteq V \times V$, while $n=|V(G)|$ and $m=|E(G)|$ denote the number of vertices and edges, resp., and $\omega:E\rightarrow \mathbb{R}_{\geq 0}$ is a weight function.
In the remainder of the paper, we use the terms \emph{graph} and \emph{digraph} interchangeably. Moreover, we simply write $V$ or $E$ when the graph $G$ is clear from the context. For simplicity, we also assume each vertex is associated with a unique identifier in $\{0,1,\dots,n-1\}$ so that each edge is uniquely identified by a pair of integers in $\{0,1,\dots,n-1\}$. We assume the graph contains neither multi-edges nor self-loops.
For a graph $G=(V,E)$ and a vertex $v\in V$, we denote by $\nout_G(v)=\{u \in V \mid (v,u) \in E\}$ and $\nin_G(v)=\{u \in V \mid (u,v) \in E\}$ the sets of outgoing and incoming neighbors of $v$, respectively, in $G$.
Given an integer $q\ge 1$, we use $[q]:=\{1,2,\ldots,q\}$.
A \emph{path} $P=(v_0 \equiv s, \dots, v_{\eta-1} \equiv t)$ in $G$ between its \emph{endpoints} $s,t$ is a sequence of $\eta$ vertices with $(v_i,v_{i+1})\in E$ for all $i$; we write $V(P)$, $E(P)$ for its vertex and edge sets and $|P|=\eta$ for its number of vertices.
A path $P$ is said to be \emph{simple} (or \emph{loopless}) if it does not contain repeated vertices.
Given a collection $\mathcal{P}$ of paths, we denote by $L = \sum_{P \in \mathcal{P}} |P|$ its \emph{total length}; note that the natural representation of $\mathcal{P}$ as a set of vertex sequences has size $\Theta(L)$ and that, for a collection of $k$ simple paths, $L \leq kn$.
The \emph{weight} of a path $P$ is defined as $\wg{P}=\sum_{i=0}^{\eta-2} \omega(v_i,v_{i+1})$.
A \emph{shortest path} between two vertices $s,t\in V$ is a path from $s$ to $t$ with minimum weight among all such paths.
In unweighted graphs, the weight of a path $P$ equals its \emph{length}, i.e., the number $|E(P)|$ of its edges.

Given a graph $G=(V,E,\omega)$ and a subset $V'\subseteq V$, we write $G\setminus V'$ for the subgraph of $G$ obtained by removing the vertices in $V'$ together with all their incident edges, i.e., the graph $(V\setminus V',E \setminus \{(u,v)\in E \mid u \in V' \lor v \in V'\},\omega)$.
Similarly, given a subset $E'\subseteq E$, we write $G\setminus E'$ for the graph $(V,E \setminus E',\omega)$.
Given two simple paths $P'=(v_0, v_1, \dots, v_{\eta-1})$ and $P''=(u_0, u_1, \dots, u_{\zeta-1})$ such that $v_{\eta-1} \equiv u_0$, we denote by $P'\oplus P''$ their \emph{concatenation}, defined as the path $(v_0, v_1, \dots, v_{\eta-1} \equiv u_0, u_1, \dots, u_{\zeta-1})$.
Note that, even when $P'$ and $P''$ are both simple, their concatenation $P'\oplus P''$ need not be simple, since the two paths may share vertices other than the junction $v_{\eta-1}\equiv u_0$.
If $P=P'\oplus Q$, with $Q$ possibly empty, we call $P'$ a \emph{prefix} of $P$. Symmetrically, if $P=Q \oplus P'$, with $Q$ possibly empty, we call $P'$ a \emph{suffix} of $P$.
Given two paths $P',P''$ we say $\Pi$ is a prefix shared by the two whenever $P'=\Pi\oplus S'$ and $P''=\Pi\oplus S''$ for some suffixes $S',S''$.
We say $\Pi$ is the longest prefix shared by two paths whenever it is maximal in terms of number of vertices.
A directed graph $T=(V_T,E_T)$ is an \emph{arborescence} rooted at a distinguished vertex $r\in V_T$, called the \emph{root}, if for every vertex $v\in V_T\setminus\{r\}$ there exists exactly one directed path from $r$ to $v$ in $T$; equivalently, $r$ has in-degree zero, every other vertex has in-degree exactly one, i.e.\ $|\nin_T(r)|=0$ and $|\nin_T(v)|=1$ for each $v\in V_T\setminus\{r\}$, and $T$ is connected when edge orientations are disregarded. As customary, we call \emph{leaves} the vertices of $T$ with out-degree zero, \emph{children} of a vertex $v\in V_T$ the vertices in $\nout_T(v)$, and \emph{depth} of $v$ the number of edges of the unique $r$--$v$ path in $T$.

The \textit{\spksisp} problem (\spksispshort) is defined as follows.
\begin{definition}
Given a graph $G=(V,E,\omega)$, two vertices $s,t \in V$, and an integer $k>0$, compute a collection
$\sol=\{P_1,P_2,\dots,P_{q}\}$ of $q\leq k$ paths from $s$ to $t$ such that:
1) each $P_i$ is simple; 2) $P_i \neq P_j$ for all $i\neq j$; 3) for any \emph{simple} path $P'$ from $s$ to $t$ not in $\sol$ and for any $P_i\in\sol$, it holds $\wg{P'}\geq \wg{P_i}$; 4) the collection is maximal, i.e., there does not exist a collection of size $q'>q$ satisfying the above conditions.
\label{def:spksisp}
\end{definition}
By Definition~\ref{def:spksisp}, if fewer than $k$ simple paths exist from $s$ to $t$, then \spksispshort requires computing all of them.
Any collection satisfying the above conditions is called a collection of \emph{top-$k$ simple shortest paths} for the pair $(s,t)$; notice that multiple feasible collections may exist, for instance when several paths have the same weight.
\blocco{Algorithms for \spksispshort}
The first polynomial-time strategy for \spksispshort was Yen's algorithm (\yenalgo) in 1971~\cite{Yen71}, and over the following decades several practical improvements were introduced. The most relevant for our purposes are the \emph{Node Classification} (\nc) algorithm~\cite{Gao10}, its enhanced variant \emph{Postponed NC} (\pnc)~\cite{ZoobiCN23},  and the recent \peek framework~\cite{PeeK}.

Given a query between two vertices $s$ and $t$, \yenalgo computes the shortest $s$--$t$ path and then repeatedly extracts the best candidate from a priority queue; from each extracted path it generates new candidates by taking each of its vertices, from the deviation point onward, as a \emph{spur node}, forbidding the edges and vertices that would yield duplicate or non-simple paths, computing a shortest path from the spur node to $t$, and concatenating it with the corresponding prefix. It terminates after $k$ paths have been reported. Its running time is $\bigO(kn\cdot\textsc{sp}(n,m))$, where $\textsc{sp}(n,m)$ denotes the cost of a single-pair shortest-path computation: $\bigO(m+n)$ via breadth-first search on unweighted graphs, and $\bigO(m+n\log n)$ via Dijkstra's algorithm on weighted ones.

The \emph{Postponed Node Classification} (\pnc)~\cite{ZoobiCN23} enhances \nc and, like it, uses residual edge weights to generate detours more efficiently than \yenalgo, avoiding many repeated shortest-path computations. It obtains each candidate detour by following a precomputed shortest-path tree towards $t$ after a single deviation from a solution path, and uses interval labels on the tree to decide in constant time whether the detour is simple; when it is not, instead of repairing it immediately with a costly shortest-path computation, \pnc enqueues it with a lower-bound key derived from the residual weights and postpones that computation until, and only if, the candidate reaches the top of the queue. This avoids many shortest-path computations, making \pnc faster than \yenalgo in several practical settings while matching its worst-case complexity. The method of \cite{KM16} has been shown to be comparable to \pnc while using substantially and impractically more memory~\cite{ZoobiCN23}.

The \peek framework~\cite{PeeK}, a recent approach for the $k$ shortest simple paths problem, combines graph reduction, graph compaction, and parallel path enumeration; we adopt only its graph-reduction phase, which is independent of the enumeration algorithm and can thus serve as a preprocessing step for \yenalgo or \pnc. Using an upper bound on the weight of the $k$-th shortest path, together with shortest-path distances from $s$ and to $t$, it removes every vertex and edge that cannot lie on any top-$k$ path, yielding a smaller graph without altering the output; the quality of the bound affects only the amount of reduction, not correctness. Its remaining components, compaction and parallel enumeration, are not used in this work.

\blocco{Pseudo-Tree Representation}
Given a graph $G$ and a collection $\mathcal{P}$ of simple paths connecting two vertices $s,t$ in $G$, the \textit{pseudo-tree}~\cite{DBLP:journals/4or/MartinsP03} of $\mathcal{P}$ is a directed graph that stores the paths in $\mathcal{P}$ in a tree-based structure where prefixes shared by multiple paths are represented only once, defined as follows.
\begin{definition}
\label{def:pseudo-tree}
Let $G=(V,E,\omega)$ be a weighted digraph with $n = |V|$ and let $\mathcal{P} = \{P_1, P_2, \ldots, P_k\}$ be a collection of $k$ pairwise distinct simple $s$--$t$ paths in $G$.
A \emph{pseudo-tree} for $\mathcal{P}$ is an arborescence $T=(V_T,E_T)$ rooted at a vertex $\text{root}(T)\in V_T$ together with a labeling function $\ell: V_T \to V$ such that:
\begin{enumerate}
  \item $\ell(\text{root}(T)) = s$;
  \item for every $P_i = (v_0 = s, v_1, \dots, v_{|P_i|-1} = t) \in \mathcal{P}$ there exists a unique root-to-leaf path
  $(x_0 = \text{root}(T), x_1, \dots, x_{|P_i|-1})$ in $T$ with $\ell(x_j) = v_j$ for all $j \in \{0, \ldots, |P_i|-1\}$;
  \item (prefix uniqueness) for every label sequence $(v_0 = s, \dots, v_j)$ that is a prefix of at least one path in $\mathcal{P}$,
  there exists exactly one vertex $x \in V_T$ at depth $j$ such that the labels along the root-to-$x$ path are $(v_0, \dots, v_j)$.
\end{enumerate}
\end{definition}
Intuitively, $T$ is obtained by overlaying the paths of $\mathcal{P}$ and merging their shared prefixes: every path is read off $T$ along a root-to-leaf path, and every vertex of $T$ corresponds to a distinct prefix of some path. The space complexity of the pseudo-tree is the same as that of $\mathcal{P}$ (specifically, it is exactly equal when paths  share no prefixes). However, in practice, this representation is convenient both for reducing space occupancy and to design fast algorithms that update solutions to \spksispshort, as we show later in the paper.
An example of pseudo-tree is illustrated in Figure~\ref{fig:pseudo-tree-sol}.
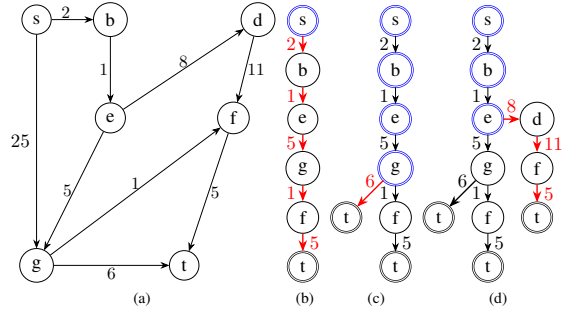
\begin{figure}[t]
\centering
\scalebox{.65}{
\subfloat[\label{fig:graph}]{%
\begin{tikzpicture}[
>=Stealth,
v/.style={draw,circle,minimum width=0.4cm},
w/.style={font=\small,inner sep=1pt}
]
\node[v] (s) at (0,5) {s};
\node[v] (b) at (1.5,5) {b};
\node[v] (e) at (1.5,3) {e};
\node[v] (d) at (4.5,5) {d};
\node[v] (g) at (0,0) {g};
\node[v] (f) at (4,3) {f};
\node[v] (t) at (3,0) {t};
\draw[->] (s) -- (b) node[w,midway,above left] {$2$};
\draw[->] (b) -- (e) node[w,midway,left] {$1$};
\draw[->] (s) -- (g) node[midway,left]  {$25$};
\draw[->] (e) -- (g) node[w,midway,left] {$5$};
\draw[->] (e) -- (d) node[w,midway,above] {$8$};
\draw[->] (d) -- (f) node[w,midway,right] {$11$};
\draw[->] (g) -- (f) node[w,midway,above] {$1$};
\draw[->] (g) -- (t) node[w,midway,below] {$6$};
\draw[->] (f) -- (t) node[w,midway,right] {$5$};
\end{tikzpicture}
}
\subfloat[\label{fig:spsol_1}]{%
\begin{tikzpicture}[
    >=Stealth,
    vertex/.style={draw,circle,minimum width=0.4cm},
]
\node[draw=blue,double,circle,minimum width=0.4cm] (s) at (0,5) {s};
\node[vertex] (b) at (0,4) {b};
\node[vertex] (e) at (0,3) {e};
\node[vertex] (g) at (0,2) {g};
\node[vertex] (f) at (0,1) {f};
\node[vertex,double] (t) at (0,0) {t};
\draw[->,red,thick] (s) -- (b) node[midway,left]  {$2$};
\draw[->,red,thick] (b) -- (e) node[midway,left]  {$1$};
\draw[->,red,thick] (e) -- (g) node[midway,left]  {$5$};
\draw[->,red,thick] (g) -- (f) node[midway,left]  {$1$};
\draw[->,red,thick] (f) -- (t) node[midway,right] {$5$};
\end{tikzpicture}
}
\subfloat[\label{fig:spsol_2}]{%
\begin{tikzpicture}[
    >=Stealth,
    vertex/.style={draw,circle,minimum width=0.4cm},
]
\node[draw=blue,double,circle,minimum width=0.4cm] (s) at (0,5) {s};
\node[draw=blue,double,circle,minimum width=0.4cm] (b) at (0,4) {b};
\node[draw=blue,double,circle,minimum width=0.4cm] (e) at (0,3) {e};
\node[draw=blue,double,circle,minimum width=0.4cm] (g) at (0,2) {g};
\node[vertex] (f) at (0,1) {f};
\node[vertex,double] (t) at (0,0) {t};
\node[vertex,double] (tL) at (-1,1) {t};
\draw[->] (s) -- (b) node[midway,left]  {$2$};
\draw[->] (b) -- (e) node[midway,left]  {$1$};
\draw[->] (e) -- (g) node[midway,left]  {$5$};
\draw[->] (g) -- (f) node[midway,left]  {$1$};
\draw[->] (f) -- (t) node[midway,right] {$5$};
\draw[->,red,thick] (g) -- (tL) node[midway,above] {$6$};
 
\end{tikzpicture}
}
\subfloat[\label{fig:spsol_3}]{%
\begin{tikzpicture}[
    >=Stealth,
    vertex/.style={draw,circle,minimum width=0.4cm},
]
\node[draw=blue,double,circle,minimum width=0.4cm] (s) at (0,5) {s};
\node[draw=blue,double,circle,minimum width=0.4cm] (b) at (0,4) {b};
\node[draw=blue,double,circle,minimum width=0.4cm] (e) at (0,3) {e};
\node[vertex] (g) at (0,2) {g};
\node[vertex] (f) at (0,1) {f};
\node[vertex,double] (t) at (0,0) {t};
\node[vertex,double] (tL) at (-1,1) {t};
\node[vertex] (d)  at (1,3) {d};
\node[vertex] (fR) at (1,2) {f};
\node[vertex,double] (tR) at (1,1) {t};
 
\draw[->] (s) -- (b) node[midway,left]  {$2$};
\draw[->] (b) -- (e) node[midway,left]  {$1$};
\draw[->] (e) -- (g) node[midway,left]  {$5$};
\draw[->] (g) -- (f) node[midway,left]  {$1$};
\draw[->] (f) -- (t) node[midway,right] {$5$};
 
\draw[->,thick] (g) -- (tL) node[midway,above] {$6$};
\draw[->,red,thick] (e) -- (d)  node[midway,above] {$8$};
\draw[->,red,thick] (d) -- (fR) node[midway,right] {$11$};
\draw[->,red,thick] (fR) -- (tR) node[midway,right] {$5$};
\end{tikzpicture}
}
}
\caption{\protect\subref{fig:graph}: a graph $G$; \protect\subref{fig:spsol_1}--\protect\subref{fig:spsol_3}: the construction of a pseudo-tree $T$ representing a collection of top-$3$ simple shortest paths from $s$ to $t$ in $G$; the three longest prefixes are highlighted in blue, the corresponding newly added branches in red; all leaves have label equal to $t$.}
\label{fig:pseudo-tree-sol}
\end{figure}
\section{New Algorithms}
\label{sec:algorithms}
In this section we present our algorithmic contributions: we show how to build a pseudo-tree from an unstructured collection of paths by adapting the classical trie-insertion procedure to our setting and by designing a worst-case-efficient alternative. Then, building on this representation, we introduce \dynk, an algorithm that updates a solution to \spksispshort when the parameter $k$ changes: given a top-$k$ collection and a new value $k'$, it returns a top-$k'$ collection, starting from the previously computed paths rather than recomputing the solution from scratch.
\subsection{Pseudo-Tree Construction} 
In the literature~\cite{DBLP:journals/4or/MartinsP03,Gao10,SSTK12}, the pseudo-tree was built \emph{incrementally}, one path at a time: starting from an empty tree, each newly found path is inserted by appending a single branch beneath an already-present prefix. This is computationally light and natural for \spksispshort, since the paths of a solution are produced as \emph{detours}, i.e. any new path coincides with a prefix of a previously found one, up to a deviation vertex $v$, and then branches off. Knowing the deviation vertex, and hence the shared prefix, the insertion reduces to hanging the suffix of the new path under the corresponding node of the tree. This construction, however, is only viable when the deviation structure is known, i.e., when one knows, for each path, of which path it is a detour and at which vertex.
Building the pseudo-tree of an \emph{arbitrary} collection of $s$--$t$ paths (for instance, a whole solution to \spksispshort given as a plain set of paths)
without such information has not been addressed in this line of work. As we observe in Remark~\ref{rem:trie}, this task amounts to constructing a \emph{trie}; accordingly, our first construction (Algorithm~\ref{alg:pseudo-tree}) is not a new algorithm but an adaptation of the classical trie-insertion procedure~\cite{Fredkin60} to our specific setting, in which the strings to be inserted are simple $s$--$t$ paths; our contribution here consists in making this connection explicit and in analyzing correctness and complexity of the resulting procedure, whose $O(kn)$ bound holds in expectation due to hash-based child lookups. We then complement it with Algorithm~\ref{alg:pseudo-tree-wc}, a construction that attains the $O(kn)$ bound deterministically, i.e., in the \emph{worst case}, and hence is asymptotically optimal with respect to running time, as the total length of the input collection can be $\Theta(kn)$. Such a from-scratch construction is precisely what our dynamic algorithms require, as they build and manipulate pseudo-trees of solutions that are given as a whole.
\begin{remark}\label{rem:trie}
Viewing each path as the string of its vertex ids, the pseudo-tree of
$\mathcal{P}$ is precisely the \emph{trie}~\cite{Fredkin60} of
the corresponding set of strings over the alphabet $V$. Two features distinguish these strings from those handled by general tries. First, since the paths are simple, each string consists of pairwise distinct symbols, whereas arbitrary strings may repeat them; in particular, string lengths are bounded by $n$. Second, the
structure of simple $s$--$t$ paths makes the collection \emph{prefix-free} (no such path can be a proper prefix of another, as this would require
revisiting $t$), so the trie needs no end-of-string markers and its leaves
are exactly the $k$ vertices labeled $t$.
\end{remark}
Informally, a pseudo-tree can be constructed by creating, for each path $P_i \in \mathcal{P}$, a directed path in the arborescence with labels corresponding to the vertices of $P_i$, while taking care not to add more than once any shared prefix.
More formally, given a graph $G$ and a set of $s$-$t$ paths $\mathcal{P}$, a pseudo-tree can be constructed as shown in Algorithm~\ref{alg:pseudo-tree}, which instantiates the textbook trie-insertion procedure on vertex-labeled paths: the children of each pseudo-tree vertex are stored in a hash map keyed by their labels, so that the child lookup at line~\ref{line:lookup} takes $\bigO(1)$ expected time. Each path is thus inserted by walking down $T$ while its prefix is already represented and branching off, with new vertices, where it is not.
\begin{algorithm}[t]\small
\KwIn{Weighted graph $G = (V, E, \omega)$ with $n = |V|$; list of $k$ simple $s$--$t$ paths $\mathcal{P} = \{P_1, P_2, \ldots, P_k\}$ on $G$}
\KwOut{Pseudo-tree $T = (V_T, E_T)$ with labeling $\ell: V_T \to V$, rooted at vertex $r~:~\ell(r) = s$}
 
$V_T \gets \{r\}$; $E_T \gets \emptyset$; $\ell(r) \gets s$\;\label{line:init:one}
$\text{root}(T) \gets r$\label{line:init:two}\;
 
\ForEach(\label{line:outer}){path $P_i = (v_0^i = s, v_1^i, \ldots, v_{|P_i|-1}^i = t) \in \mathcal{P}$}{
    $u \gets \text{root}(T)$\;
    \For(\label{line:inner}){$j \gets 1$ \KwTo $|P_i| - 1$}{
        $v \gets v_j^i$\;
        \lIf(\label{line:lookup}){$\exists$ child $x$ of $u$ in $T$ : $\ell(x) = v$}{
            $u \gets x$}
        \Else{
            Create new vertex $x$ with $\ell(x) \gets v$\;\label{line:insert}
            $V_T \gets V_T \cup \{x\}$; $E_T \gets E_T \cup \{(u, x)\}$\;
            $u \gets x$\;
        }
    }
}
\Return{$T = (V_T, E_T)$ with labeling $\ell$}\;
\caption{\build: Pseudo-Tree Construction}\label{alg:pseudo-tree}
\end{algorithm}
\begin{algorithm}[t]
\small
\KwIn{Weighted graph $G = (V, E, \omega)$ with $n = |V|$; list of $k$ simple $s$--$t$ paths $\mathcal{P} = \{P_1, P_2, \ldots, P_k\}$ on $G$, where $P_i = (v_0^i = s, v_1^i, \ldots, v_{|P_i|-1}^i = t)$}
\KwOut{Pseudo-tree $T = (V_T, E_T)$ with labeling $\ell: V_T \to V$, rooted at vertex $r~:~\ell(r) = s$}
$\textit{seen}[0..n{-}1] \gets \textbf{false}$\;
$\textit{node}[0..n{-}1] \gets \textbf{null}$\;
$\textit{groups}[0..n{-}1] \gets \text{empty stacks}$\;
$V_T \gets \{r\}$; $E_T \gets \emptyset$; $\ell(r) \gets s$\;
$\text{root}(T) \gets r$\;
\lForEach{$i \in [k]$}{
    $\textit{groups}[s].\text{push}((i, 1, r))$}
$\addrecur(r, s)$\;
\Return{$T = (V_T, E_T)$ with labeling $\ell$}\;
\caption{\buildwc: Worst-Case Pseudo-Tree Construction}\label{alg:pseudo-tree-wc}
\end{algorithm}
{\renewcommand{\algorithmcfname}{Procedure}%
\begin{algorithm}[t]
\small
\caption{\addrecur($\bar{u}$, $u$): Add children of pseudo-tree vertex $\bar{u}$ for active path segments in $\textit{groups}[u]$}\label{proc:add}
\KwIn{Pseudo-tree vertex $\bar{u} \in V_T$; graph vertex $u \in V$ : $\ell(\bar{u}) = u$}
\lIf{$\textit{groups}[u] = \emptyset$}{\Return}
$\textit{touched} \gets$ empty stack\;
\While(\label{line:phase:1}){$(i, d, w) \gets \textit{groups}[u].\text{top}()$}{
    \lIf{$w \neq \bar{u}$}{\textbf{break}\label{line:triple:break}}
    $(i, d, w) \gets \textit{groups}[u].\text{pop}()$\;
    $v \gets v_d^i$\;
    \If{\textbf{not} $\textit{seen}[v]$}{
        $\textit{seen}[v] \gets \textbf{true}$\;
        Create new pseudo-tree vertex $x$ with $\ell(x) \gets v$\;
        $V_T \gets V_T \cup \{x\}$; $E_T \gets E_T \cup \{(\bar{u}, x)\}$\;
        $\textit{node}[v] \gets x$\;
        $\textit{touched}.\text{push}((v,x))$\;
    }
    \If{$d < |P_i| - 1$}{
        $x \gets \textit{node}[v]$\;
        $\textit{groups}[v].\text{push}((i, d+1, x))$\;\label{line:phase:push}
    }
}
\ForEach(\label{line:phase:2}){$(v, x) \in \textit{touched}$}{
    $\textit{seen}[v] \gets \textbf{false}$;
    $\textit{node}[v] \gets \textbf{null}$\;
}
\While(\label{line:phase:3}){$\textit{touched} \neq \emptyset$}{
    $(v, x) \gets \textit{touched}.\text{pop}()$\;
    $\addrecur(x, v)$\;
}
\end{algorithm}}
For simplicity, we say that we add a vertex $v$ of $G$ to $T$ to indicate that we add a vertex $x$ to $V_T$ with label $\ell(x) = v$. Similarly, we say that there exists a path $P = (v_0, v_1, \dots)$ in a pseudo-tree whenever the pseudo-tree admits a root-to-leaf path $(x_0, x_1, \dots)$ whose labels are $(\ell(x_0), \ell(x_1), \dots) = (v_0, v_1, \dots)$.
We prove the following results.
\begin{theorem}[Correctness]\label{thm:pseudo-tree-correctness}
Let $G = (V, E, \omega)$ be a weighted graph and let $\mathcal{P} = \{P_1, \ldots, P_k\}$ be a collection of pairwise distinct simple $s$--$t$ paths on $G$. Algorithm~\ref{alg:pseudo-tree} correctly constructs a pseudo-tree $T$ for $\mathcal{P}$ satisfying Definition~\ref{def:pseudo-tree}.
\end{theorem}
\begin{proof}
We first show that, throughout the execution, $T$ is an arborescence rooted at $r$ and that, for every vertex $u \in V_T$, no two children of $u$ carry the same label. Both properties hold for the initial tree, which consists of the root only. They are preserved by every insertion (line~\ref{line:insert}): a new vertex $x$ is added together with the single edge $(u,x)$, where $u$ is an existing vertex of $T$, so $x$ has in-degree one and is reachable from $r$ through a unique directed path (the one reaching $u$, extended by $(u,x)$); moreover, $x$ is created only when the lookup at line~\ref{line:lookup} fails, i.e., when $u$ has no child labeled $\ell(x)$, so the labels of the children of $u$ remain pairwise distinct.

Next, we prove by induction on the iterations of the inner loop (line~\ref{line:inner}) the following invariant: at the beginning of the iteration with index $j$ of the inner loop for path $P_i = (v_0^i = s, v_1^i, \ldots, v_{|P_i|-1}^i = t)$, the current vertex $u$ is a vertex of $T$ at depth $j-1$ whose root-to-$u$ label sequence is $(v_0^i, \ldots, v_{j-1}^i)$. The invariant holds for $j=1$, since $u = \text{root}(T)$ has depth $0$ and $\ell(\text{root}(T)) = s = v_0^i$. Assume it holds at the beginning of iteration $j$. During the iteration, the algorithm either moves to the (unique, by the above property) child $x$ of $u$ with $\ell(x) = v_j^i$, if it exists, or creates such a child and moves to it. In both cases, at the end of the iteration $u$ is a vertex at depth $j$ whose root-to-$u$ label sequence is $(v_0^i, \ldots, v_j^i)$, which establishes the invariant for $j+1$.

By the invariant, after the inner loop for $P_i$ has completed, $T$ contains a root-to-$x$ path whose labels are exactly $(v_0^i, \ldots, v_{|P_i|-1}^i)$, i.e., a path representing $P_i$; since, by Remark~\ref{rem:trie}, no path in $\mathcal{P}$ is a proper prefix of another, the endpoint of such a path, labeled $t$, is a leaf of the final tree. Hence condition~2 of Definition~\ref{def:pseudo-tree} is satisfied, with uniqueness of the representing path following from the pairwise distinctness of sibling labels; condition~1 holds by initialization (lines~\ref{line:init:one}--\ref{line:init:two}).

Finally, we prove condition~3 (prefix uniqueness). Existence: every label sequence $(v_0 = s, \ldots, v_j)$ that is a prefix of some $P_i \in \mathcal{P}$ is realized by a vertex at depth $j$, namely the vertex reached by the algorithm at the end of iteration $j$ of the inner loop for $P_i$. Uniqueness: suppose, towards a contradiction, that two distinct vertices $x \neq y$ at depth $j$ realize the same label sequence, and choose a pair minimizing $j$. Then $j \geq 1$ and, by minimality, the parents of $x$ and $y$ realize the same label sequence of length $j$, hence coincide; but then $x$ and $y$ are two distinct children of the same vertex carrying the same label $v_j$, contradicting the pairwise distinctness of sibling labels. Therefore, the output of Algorithm~\ref{alg:pseudo-tree} satisfies all conditions of Definition~\ref{def:pseudo-tree}.
\end{proof}
\begin{theorem}[Time Complexity]\label{thm:pseudo-tree-time}
Under the same hypotheses, assuming the children of each $u \in V_T$ are stored in a hash map keyed by their labels, Algorithm~\ref{alg:pseudo-tree} constructs $T = (V_T, E_T)$ with labeling $\ell: V_T \to V$ in $O(L)$ expected time.
\end{theorem}
\begin{proof}
Since each path in $\mathcal{P}$ is simple, it contains at most $n$ vertices, i.e., $|P_i| \leq n$ for all $i \in [k]$.
The initialization (lines \ref{line:init:one}--\ref{line:init:two}), which creates the tree with a single root vertex $r$ and sets $\ell(r) = s$, takes $O(1)$ time. The outer loop (line \ref{line:outer}) iterates over all $k$ paths, and for each path $P_i$ the inner loop (line \ref{line:inner}) iterates over at most $n - 1$ vertices (excluding the source $s$). Each inner iteration performs a child lookup (line \ref{line:lookup}), i.e., checking whether the current vertex $u$ has a child $x$ with $\ell(x) = v$, which takes $O(1)$ expected time using a hash map that associates each vertex $u \in V_T$ with its children indexed by label, and, when the prefix diverges, an insertion (line \ref{line:insert}), i.e., creating a new vertex $x$, setting $\ell(x) = v$ and adding the edge $(u, x)$ to $E_T$, which takes $O(1)$ worst-case time.
The total number of inner loop iterations, summed over all paths, is $\sum_{i=1}^{k}(|P_i|-1) < L$. Since each iteration has $O(1)$ expected cost (dominated by the child lookup), the total running time is $O(L)$ expected; as $L \leq kn$, this is in particular $O(kn)$.
\end{proof}
\begin{theorem}[Space Complexity]\label{thm:pseudo-tree-space}
Under the same hypotheses, Algorithm~\ref{alg:pseudo-tree} requires $O(L)$ space; in particular, $|V_T| \leq L$.
\end{theorem}
\begin{proof}
The space usage is dominated by the pseudo-tree structure $T = (V_T, E_T)$: in the worst case, no path shares a non-empty prefix with any other. The total number of vertices in $V_T$ is at most $|V_T| = 1 + \sum_{i=1}^{k} (|P_i| - 1) \leq L$, where the term $1$ corresponds to the root $r$, counted only once.
Moreover, since $T$ is a tree, the number of edges is $|E_T| = |V_T| - 1 = O(L)$.
The labeling function $\ell$ requires $O(|V_T|) = O(L)$ space.
Finally, for each vertex $u \in V_T$, we maintain a hash map of its children keyed by label. The total number of entries across all hash maps equals the number of edges, i.e., $O(L)$.
Summing all contributions, the total space required is $O(L)$; since each path is simple, $|P_i| \leq n$ for all $i$, hence $L \leq kn$.
\end{proof}
Note that $\mathcal{P}$ is recovered from $T$ by any traversal from the root, each root-to-leaf label sequence inducing one path of $\mathcal{P}$.
In the rest of the paper, we  write $T$ to indicate the pseudo-tree $T(\mathcal{P})$ of a collection $\mathcal{P}$ when clear from the context.
\blocco{Worst-Case Pseudo-Tree Construction}
Algorithm~\ref{alg:pseudo-tree} achieves running time linear in the input size, but only in an expected sense, due to hash map operations. We now present an algorithm, Algorithm~\ref{alg:pseudo-tree-wc}, that achieves $O(n + L)$ \emph{worst-case} time by replacing hash-based lookups with direct array indexing: its running time is linear in the input size up to an additive $\Theta(n)$ term, needed to initialize its direct-address arrays, and is therefore asymptotically optimal whenever $L = \Omega(n)$. As we show in Section~\ref{sec:experiments}, however, the stronger theoretical guarantee does not translate into practical gains: on realistic instances the input paths are short relative to the graph size, i.e., $L \ll n$, the additive $\Theta(n)$ term dominates, and the simpler trie-based construction is consistently faster.
The key idea is to process paths \emph{level-by-level} rather than path-by-path, using a depth-first traversal of the pseudo-tree under construction. We maintain \emph{active path segments}, where each segment is represented as a triple $(i, d, w)$ indicating that path $P_i$ is currently being processed at position $d$, and the segment belongs to the subtree rooted at pseudo-tree vertex $w$. The tagging with $w$ is essential: since the same graph vertex can appear at different depths across different paths, we must distinguish which segments belong to the current recursive call.
 
The algorithm uses three global arrays indexed by graph vertices in $\{0, \ldots, n-1\}$: a boolean array $\textit{seen}$, tracking which vertices were encountered at the current depth; an array $\textit{node}$, mapping each such vertex to its current pseudo-tree vertex (for $O(1)$ retrieval when tagging continuing segments); and an array $\textit{groups}$ of stacks holding the segments of each vertex. The LIFO order of $\textit{groups}$ is crucial: as we recurse depth-first, segments of deeper calls, pushed later, must be processed before those of ancestor calls.
 
The recursive procedure $\addrecur(\bar{u}, u)$ (Procedure~\ref{proc:add}), with $u = \ell(\bar{u})$, processes the segments of $\textit{groups}[u]$ belonging to $\bar{u}$. In Phase~1 (line~\ref{line:phase:1}) it pops the top triple $(i, d, w)$ while $w = \bar{u}$ (stopping otherwise, line~\ref{line:triple:break}, as the triple then belongs to an ancestor call) and examines $v = v_d^i$: if $v$ is unseen at this depth, it creates a child $x$ with $\ell(x)=v$, sets $\textit{node}[v]=x$, and records $(v,x)$ in the local stack $\textit{touched}$; if $P_i$ continues, it pushes $(i, d+1, x)$ onto $\textit{groups}[v]$ (line~\ref{line:phase:push}), tagged with $x$.
After processing all relevant segments, Phase~2 (line~\ref{line:phase:2}) resets $\textit{seen}[v]$ and $\textit{node}[v]$ for each touched vertex $v$. This reset occurs \emph{before} any recursive call, ensuring that each recursive invocation sees clean global arrays. Finally, Phase~3 (line~\ref{line:phase:3}) recurses on each child: for each pair $(v, x)$ in $\textit{touched}$, it calls $\addrecur(x, v)$. The recursive call processes only the triples in $\textit{groups}[v]$ that are tagged with $x$, stopping when it encounters a triple tagged with a different vertex.
\begin{remark}\label{rem:buildwc-scope}
The bookkeeping of \buildwc is deliberately tailored to our setting, trading generality for efficiency. Its stack-plus-tag discipline is sound as long as no input string contains two equal \emph{adjacent} symbols: were $v_d^i = v_{d-1}^i$ for some string, Phase~1 would push continuations onto the very stack it is draining, and the tag check would prematurely abandon the segments left beneath them. Simple paths satisfy this requirement a fortiori, as they repeat no vertex at all (in fact, the algorithm applies verbatim to any prefix-free collection of pairwise distinct walks in a graph with no self-loops), whereas fully general strings would require materializing the partition of the active segments before recursing, at the cost of allocating fresh containers at every call. By pushing continuations directly into the preallocated global stacks, \buildwc avoids such allocations altogether; we expect this restriction to translate into smaller constant factors in the construction time, which, while leaving the asymptotic bound untouched, do matter in practice.
\end{remark}
We can prove the following results.
\begin{theorem}[Correctness]
Let $G = (V, E, \omega)$ be a weighted graph and let $\mathcal{P} = \{P_1, \ldots, P_k\}$ be a collection of pairwise distinct simple $s$--$t$ paths on $G$. Algorithm~\ref{alg:pseudo-tree-wc} correctly constructs a pseudo-tree $T$ for $\mathcal{P}$ satisfying Definition~\ref{def:pseudo-tree}.
\end{theorem}
\begin{proof}
We prove correctness by establishing the following invariants:
\textbf{Invariant 1:} Each triple $(i, d, w)$ is processed exactly once, in the call $\addrecur(w, \ell(w))$.
This holds because: (a) each triple is tagged with the pseudo-tree vertex $w$ to which it belongs; (b) the check $w \neq \bar{u}$ ensures that a triple is only processed when the current call matches its tag (if the tag does not match, the triple is not popped and the loop terminates); (c) once processed, the triple is popped and never inspected again by the same call.
 
\textbf{Invariant 2:} When $\addrecur(\bar{u}, u)$ is called, all triples in $\textit{groups}[u]$ tagged with $\bar{u}$ are consecutive at the top of the stack.
This holds because: (a) the LIFO discipline ensures that triples pushed later are at the top; (b) recursive calls process depth-first, so a child's triples are pushed onto $\textit{groups}[v]$ and processed completely before sibling branches are explored; (c) triples tagged with ancestor vertices remain deeper in the stack, untouched until all descendant calls return.
 
\textbf{Invariant 3:} The $\textit{seen}$ and $\textit{node}$ arrays are clean (respectively, all $\textbf{false}$ and all $\textbf{null}$ for the relevant entries) at the start of each call to \addrecur.
This holds because Phase~2 resets $\textit{seen}[v]$ and $\textit{node}[v]$ for all touched vertices before any recursive call in Phase~3.
 
\textbf{Construction correctness:} For each path $P_i = (v_0^i, v_1^i, \ldots, v_{|P_i|-1}^i)$, the algorithm creates a sequence of pseudo-tree vertices $x_0 = r, x_1, \ldots, x_{|P_i|-1}$ with $\ell(x_j) = v_j^i$ and edges $(x_{j-1}, x_j)$ for $j \in [|P_i|-1]$. The prefix uniqueness property is guaranteed by the $\textit{seen}$ check: if multiple paths share a prefix, they share the same pseudo-tree vertices for that prefix.
Therefore, the output satisfies all three conditions of Definition~\ref{def:pseudo-tree}.
\end{proof}
\begin{theorem}[Time Complexity]\label{thm:pseudo-tree-wc-time}
Let $G = (V, E, \omega)$ be a weighted graph with $n = |V|$ vertices, and let $\mathcal{P}$ be a list of $k$ simple $s$--$t$ paths on $G$. Algorithm~\ref{alg:pseudo-tree-wc} constructs the pseudo-tree $T = (V_T, E_T)$ with labeling $\ell: V_T \to V$ in $O(n + L)$ worst-case time.
\end{theorem}
\begin{proof}
We assume graph vertices are represented as integers in $\{0, \ldots, n-1\}$, so that array indexing takes $O(1)$ worst-case time. We also assume that stacks support $O(1)$ worst-case push, pop, and top operations. Allocating and initializing the global arrays $\textit{seen}$, $\textit{node}$, and $\textit{groups}$, each of size $n$, takes $\Theta(n)$ worst-case time; this accounts for the additive term of the claimed bound, and every remaining cost is charged to $L$ as follows.
The key observation is that each triple $(i, d, w)$ is pushed exactly once (when created in Phase~1) and popped exactly once (when processed in the matching call $\addrecur(w, \ell(w))$). The total number of triples created is $\sum_{i=1}^{k} (|P_i| - 1) < L$, since for each path we process positions $1, 2, \ldots, |P_i|-1$.
Processing a single triple in Phase~1 takes $O(1)$ worst-case time: inspecting the top of $\textit{groups}[u]$ and checking $w \neq \bar{u}$, popping the triple, accessing $v_d^i$, and testing and updating $\textit{seen}[v]$ are all $O(1)$ worst-case operations; and, when $v$ has not been seen, creating a new pseudo-tree vertex, adding the edge, setting $\textit{node}[v]$ and pushing onto $\textit{touched}$, as well as retrieving $x = \textit{node}[v]$ and pushing the continuing segment onto $\textit{groups}[v]$, are likewise $O(1)$ worst-case. Hence Phase~1 spends $O(1)$ per triple, i.e. $O(L)$ over all triples. To this we add at most one terminating inspection per call (the mismatch $w \neq \bar{u}$ that breaks the loop); since there are at most $O(L)$ calls (one per pseudo-tree vertex created, plus the initial call), these inspections cost $O(L)$ in total.
The other two phases are charged to the same quantities. In Phase~2, resetting $\textit{seen}[v]$ and $\textit{node}[v]$ takes $O(|\textit{touched}|)$ time per call, which summed over all calls equals the total number of pseudo-tree vertices created, i.e., $O(L)$; in Phase~3, each child is handled with one pop from $\textit{touched}$ and one recursive call, again $O(L)$ in total. Therefore, the total running time is $O(n + L)$ worst-case, which is in particular $O(kn)$ since $L \leq kn$ and $n \leq kn$.
\end{proof}
\begin{theorem}[Space Complexity]
Algorithm~\ref{alg:pseudo-tree-wc} requires $O(n + L)$ space.
\end{theorem}
\begin{proof}
The space usage consists of the following components: the global arrays $\textit{seen}$ and $\textit{node}$, each requiring $O(n)$ space; the global array $\textit{groups}$, requiring $O(n)$ space for the array structure itself; the total number of triples stored across all stacks in $\textit{groups}$ at any point in time is at most $O(L)$, since each triple is created exactly once and removed exactly once; the pseudo-tree structure with $|V_T| \leq L$ vertices and $|E_T| = |V_T| - 1 = O(L)$ edges; the labeling function $\ell$, taking $O(|V_T|) = O(L)$ space; the local $\textit{touched}$ stacks across all active stack frames. At any point during execution, the recursion stack has depth at most $\max_i |P_i| \leq n$. The total number of pairs across all active $\textit{touched}$ stacks is bounded by the total number of pseudo-tree vertices, i.e., $O(L)$.
Finally, we need to count the recursion stack itself: each stack frame uses $O(1)$ space for local variables (excluding the $\textit{touched}$ contents, which are already counted above), and the maximum depth is $n$, contributing $O(n)$ space.
Summing all contributions, the total space required is $O(n) + O(L) = O(n + L)$, which is in particular $O(kn)$.
\end{proof}
\subsection{A Worked Example of Worst-Case Pseudo-Tree Construction}
\label{sec:pseudo-tree-examples}
In what follows, we illustrate the execution of Algorithm~\ref{alg:pseudo-tree-wc} on an example in which the same vertex is encountered at different depths of the recursion, which shows why tagging segments with pseudo-tree vertices and processing the stacks in LIFO order are needed for correctness.
\blocco{Example}
Consider a graph with vertices $\{s, a, c, d, t\}$ and the following four simple $s$--$t$ paths:
\[
P_1 = (s, a, d, t), \quad P_2 = (s, a, c, t), \quad P_3 = (s, d, c, t), \quad P_4 = (s, d, a, t).
\]
The example is designed so that both $a$ and $d$ appear at different depths: $a$ at position $1$ in $P_1, P_2$ and position $2$ in $P_4$; $d$ at position $1$ in $P_3, P_4$ and position $2$ in $P_1$. It demonstrates the necessity of tagging triples with pseudo-tree vertices and the LIFO discipline.

After Phase~1 (Call $\addrecur(r, s)$): $\textit{touched} = [(d, x_2), (a, x_1)]$, $\textit{groups}[a] = [(2,2,x_1), (1,2,x_1)]$, $\textit{groups}[d] = [(4,2,x_2), (3,2,x_2)]$.
Phase~1 (Call $\addrecur(x_1, a)$) processes $\textit{groups}[a]$:
\begin{itemize}
    \item $(1,2,x_1)$: $v = d$. Create $x_3$ with $\ell(x_3) = d$. Push $(1,3,x_3)$ onto $\textit{groups}[d]$.
    \item $(2,2,x_1)$: $v = c$. Create $x_4$ with $\ell(x_4) = c$. Push $(2,3,x_4)$ onto $\textit{groups}[c]$.
\end{itemize}
Now $\textit{groups}[d] = [(4,2,x_2), (3,2,x_2), (1,3,x_3)]$ (top on right).
Phase~3 recurses: first on $(c, x_4)$, then on $(d, x_3)$.

Call $\addrecur(x_4, c)$ creates $x_6$ with $\ell(x_6) = t$ (leaf) while call $\addrecur(x_3, d)$, Phase~1, processes $\textit{groups}[d]$:
\begin{itemize}
    \item Inspect $(1,3,x_3)$: $w = x_3 = \bar{u}$ \checkmark. Pop and process. $v = t$. Create $x_5$ (leaf).
    \item Inspect $(3,2,x_2)$: $w = x_2 \neq x_3$. \textbf{Break} (triple remains in stack).
\end{itemize}
The triples $(3,2,x_2)$ and $(4,2,x_2)$ remain in $\textit{groups}[d]$ for the later call to $\addrecur(x_2, d)$.

Phase~1 (Call $\addrecur(x_2, d)$) processes $\textit{groups}[d]$:
\begin{itemize}
    \item $(3,2,x_2)$: $w = x_2 = \bar{u}$ \checkmark. $v = c$. Create $x_7$. Push $(3,3,x_7)$ onto $\textit{groups}[c]$.
    \item $(4,2,x_2)$: $w = x_2 = \bar{u}$ \checkmark. $v = a$. Create $x_8$. Push $(4,3,x_8)$ onto $\textit{groups}[a]$.
\end{itemize}
Phase~3: $\addrecur(x_8, a)$ creates $x_{10}$ ($\ell = t$); $\addrecur(x_7, c)$ creates $x_9$ ($\ell = t$).
Resulting pseudo-tree: both $a$ and $d$ appear twice in the pseudo-tree, reached via different prefixes. See Figure~\ref{fig:pseudo-tree-example}.
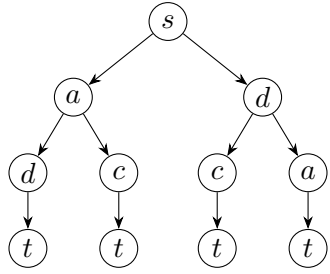
\begin{figure}[ht]
\centering
\begin{tikzpicture}[
    >=Stealth,
    level distance=1cm,
    sibling distance=1.4cm,
    every node/.style={draw, circle, minimum size=0.5cm, inner sep=1pt},
    edge from parent/.style={draw, ->},
    level 1/.style={sibling distance=2.5cm},
    level 2/.style={sibling distance=1.2cm},
]
\node {$s$}
    child {
        node {$a$}
        child {
            node {$d$}
            child { node {$t$} }
        }
        child {
            node {$c$}
            child { node {$t$} }
        }
    }
    child {
        node {$d$}
        child {
            node {$c$}
            child { node {$t$} }
        }
        child {
            node {$a$}
            child { node {$t$} }
        }
    };
\end{tikzpicture}
\caption{Pseudo-tree for $P_1 = (s,a,d,t)$, $P_2 = (s,a,c,t)$, $P_3 = (s,d,c,t)$, $P_4 = (s,d,a,t)$. The tagging mechanism and LIFO discipline ensure that triples are processed by the correct recursive call.}
\label{fig:pseudo-tree-example}
\end{figure}
\blocco{A key observation}
In the call $\addrecur(x_3, d)$, the algorithm inspects $(3,2,x_2)$ at the top of $\textit{groups}[d]$, identifies that it belongs to a different pseudo-tree vertex ($x_2 \neq x_3$), and terminates the loop without removing the triple. This triple remains available for the later call $\addrecur(x_2, d)$. Without the tagging mechanism, the algorithm would incorrectly process $(3,2,x_2)$ as part of the subtree rooted at $x_3$, leading to an invalid pseudo-tree.

\subsection{Updating Solutions to \spksispshort under Changes to \texorpdfstring{$k$}{k}}
In what follows, we present a new algorithm, named \dynk, for updating solutions to \spksispshort under changes to parameter $k$.
The algorithm, whose pseudo-code is given in Algorithm~\ref{algo:dynamic_k}, is based on representing a solution to \spksispshort, which is a collection of simple paths, in the form of a pseudo-tree. It takes as input a weighted digraph $G=(V,E,\omega)$, a pair of vertices $s,t \in V$ for which we have a collection \sol of top-$k$ simple shortest paths from $s$ to $t$, and an integer $k'\neq k$, and returns the updated collection \sol of top-$k'$ simple shortest paths from $s$ to $t$.
\begin{algorithm}[t]
\small
\SetAlgoLined
\KwIn{Weighted digraph $G=(V,E,\omega)$, collection \sol of top-\texorpdfstring{$k$}{k} simple shortest paths from $s$ to $t$, integer $k'\neq k, k'>0$.}
\KwOut{Collection \sol of top-$k'$ simple shortest paths from $s$ to $t$.}
\If{$k'<k$}{
$\sol \gets k'$ lightest paths in \sol\;
\KwRet \sol\;
}
$T \gets \build(G,\sol)$\tcp*{Build pseudo-tree (Algorithm~\ref{alg:pseudo-tree} or~\ref{alg:pseudo-tree-wc})}
$Q\gets \emptyset$; $\found \gets \emptyset$\;
\ForEach{vertex $h \in V_T$}{
$\expand(T,G,h,Q,\found)$\tcp*{Seed $Q$ (Procedure~\ref{proc:expand})}
}
\While(\label{line:dynk:while}){$Q \neq \emptyset$}{
$\langle \Pi, h, S\rangle \gets Q.dequeueMin()$; $\found \gets \found \setminus \{\Pi\}$\;
$\sol \gets \sol \cup \{\Pi\}$\;
\lIf{$|\sol|=k'$}{\KwRet \sol}
$(T, W) \gets \merge(T,h,S)$\tcp*{Merge into $T$; $W$: vertices to (re-)expand (Procedure~\ref{algo:merge_and_return})}
\lForEach{$h \in W$}{
$\expand(T,G,h,Q,\found)$}
}
\KwRet \sol
\caption{\dynk: Updating a solution to \spksispshort under changes of $k$}
\label{algo:dynamic_k}
\end{algorithm}
{\renewcommand{\algorithmcfname}{Procedure}%
{\renewcommand{\algorithmcfname}{Procedure}%
\begin{algorithm}[t]
\small
\SetAlgoLined
\KwIn{Pseudo-tree $T=(V_T,E_T)$ with labeling $\ell$; weighted digraph $G$; vertex $h\in V_T$; min-priority queue $Q$; membership set \found\ of the paths currently in $Q$.}
\KwOut{If $h$ admits a feasible deviation not already in $Q$, the corresponding candidate is enqueued in $Q$, together with the vertex $h$ it deviates from and its suffix $S$.}
$v\gets \ell(h)$\;
\lIf{$v=t$}{\KwRet}
$\Gamma \gets$ path from $\text{root}(T)$ to $h$ in $T$\;
$P\gets$ path of $G$ induced by the labels of $\Gamma$, i.e., $P = (\ell(x_0), \ldots, \ell(h))$\;
$H\gets G$\;
\lForEach{$u \in \nout_T(h)$}{
$H \gets H\setminus\{(v,\ell(u))\}$}
$H \gets H \setminus \big(V(P) \setminus \{v\}\big)$\;
$S\gets \dij(H,v,t)$\;
\If{$S\neq \emptyset$}{
$D\gets P\oplus S$\;
\If{$D \notin \found$}{
$Q.enqueue(\wg{D},\langle D, h, S\rangle)$\tcp*{keep the deviation vertex and suffix}
$\found \gets \found \cup \{D\}$\;
}
}
\caption{\expand: Generate a shortest feasible deviation for a pseudo-tree vertex}
\label{proc:expand}
\end{algorithm}}
\begin{algorithm}[ht]
\small
\SetAlgoLined
\KwIn{Pseudo-tree $T = (V_T, E_T)$ with labeling $\ell$; deviation vertex $h \in V_T$ and suffix $S = (u_0 = \ell(h), u_1, \ldots, u_r = t)$ of the extracted candidate, as recorded by \expand.}
\KwOut{Updated pseudo-tree $T$; set $W$ of vertices to be (re-)expanded.}
$x \gets h$; $i \gets 1$\;
\While(\tcp*[f]{no-op unless ties reordered extractions}){$i \leq r$ \textbf{and} $x$ has a child $y$ with $\ell(y) = u_i$}{
    $x \gets y$; $i \gets i+1$\;
}\label{line:merge:branch}
$W \gets \{h, x\}$\;
\For{$j \gets i$ \KwTo $r$}{
    Create new vertex $z$ with $\ell(z) \gets u_j$\;
    $V_T \gets V_T \cup \{z\}$; $E_T \gets E_T \cup \{(x, z)\}$\;
    $W \gets W \cup \{z\}$; $x \gets z$\;
}
\KwRet $(T, W)$\;
\caption{\merge: Append a generated candidate below its deviation vertex}
\label{algo:merge_and_return}
\end{algorithm}}
Specifically, the algorithm first handles the easy case: if fewer paths are requested ($k'<k$), it simply selects and returns the $k'$ lightest paths already contained in \sol, and terminates. 
When $k'>k$, the algorithm must extend the solution with additional simple shortest paths. It first builds a pseudo-tree $T$ from \sol (via Algorithm~\ref{alg:pseudo-tree} or Algorithm~\ref{alg:pseudo-tree-wc}), so that prefixes shared by multiple paths are represented only once, and initializes a min-priority queue $Q$ of candidate paths.
The generation of a candidate from a single pseudo-tree vertex is encapsulated in the procedure \expand (Procedure~\ref{proc:expand}). Given a vertex $h$ with $v=\ell(h)$, \expand considers the root-to-$h$ path in $T$, whose labels induce a path $P$ in $G$ ending at $v$, and builds a subgraph $H$ of $G$ that does not contain (i) all vertices of $P$ except $v$ (to enforce simplicity) and (ii) the edges $(v,\ell(h'))$ corresponding to children $h'$ of $h$ in $T$ (to avoid reproducing already found paths). In $H$ it runs a shortest-path subroutine $\dij(H,v,t)$ and, if a suffix $S$ exists, enqueues the deviation $D=P\oplus S$ with priority $\wg{D}$, unless $D$ is already in $Q$: the membership set \found, implemented as a hash set keyed on the vertex sequence of the path, keeps the queue duplicate-free at a cost of $\bigO(n)$ per test, which is dominated by the shortest-path computation just performed.
Note that any leaf labeled $t$ is a complete path and admits no deviation, so \expand returns immediately on it. The algorithm seeds $Q$ by invoking \expand on every vertex of $T$: as each vertex is a distinct prefix of a stored path (a spur candidate in Yen's terminology), this reconstructs, from the solution alone, exactly the frontier of detours a from-scratch run would have accumulated.
After seeding $Q$, the algorithm repeatedly extracts from $Q$ the minimum-weight path $\Pi$, which is the next simple shortest path. Each extracted path is added to \sol and to the pseudo-tree $T$ via \merge (Procedure~\ref{algo:merge_and_return}), which appends only the previously missing suffix as a new branch. Note that \merge performs no search for the longest represented prefix: every candidate is generated by \expand as a detour, so its deviation vertex $h$ and its suffix $S$ are known and are recorded in $Q$ alongside the candidate, and the insertion reduces to hanging $S$ below $h$ in $\bigO(|S|)$ time. The only care needed is that $T$ may have grown below $h$ between the generation of the candidate and its extraction, which can happen only among candidates of equal weight; procedure \merge therefore descends from $h$ as long as the labels of $S$ match existing children (line~\ref{line:merge:branch}), and appends the rest. Besides the updated tree, \merge returns the set $W$ to (re-)expand: the newly created suffix vertices, the vertex that acquired a new child, and $h$ itself, whose own candidate has just been consumed. If $|\sol|=k'$ the algorithm returns; otherwise it invokes \expand on every vertex of $W$, leaving all other vertices untouched, as their children, and hence their candidates in $Q$, are unchanged. In this way, the solution is grown in non-decreasing weight order while the pseudo-tree guides the search and prevents redundant exploration of shared prefixes. We prove the following results.
Observe that the membership test of \expand might be overly cautious: $T$ evolves and we do not have the guarantee that a path generated from $h$ is not generated again from a vertex created below $h$ later during the execution, with the second generation preceding the extraction of the first when the two candidates have equal weight. This is also the only way $T$ can grow below $h$ before the candidate of $h$ is extracted. Discarding the second copy loses no candidate, since extracting the copy already in $Q$ triggers a \merge\ whose branching vertex $x$ is precisely the vertex that would have re-generated it, which is therefore re-expanded.
\begin{theorem}[Correctness]
\label{thm:dynk-correctness}
Algorithm~\ref{algo:dynamic_k} computes a collection of top-$k'$ simple shortest paths.
\end{theorem}
\begin{proof}
If $k'<k$, the algorithm returns the $k'$ lightest paths among those in \sol.
Since \sol is a valid top-$k$ solution, every simple $s$--$t$ path not in \sol has weight at least that of the $k$-th path; hence the $k'$ lightest paths in \sol form a valid top-$k'$ solution (ties can be handled arbitrarily). This also covers the case in which fewer than $k$ (or $k'$) simple $s$--$t$ paths exist.
Assume now $k'>k$. The algorithm first builds the pseudo-tree $T$ representing exactly the $k$ shortest simple $s$--$t$ paths in \sol.
By construction, $T$ satisfies the prefix-uniqueness property: for every label sequence $(v_0 = s, \ldots, v_j)$ that is a prefix of at least one path in \sol, there exists exactly one vertex $h \in V_T$ whose root-to-$h$ path has that sequence of labels.
For a vertex $h$ of $T$ with $\ell(h)=v\neq t$, procedure \expand considers the $s$--$v$ path $P$ induced by the labels of the root-to-$h$ path and the subgraph $H$ obtained from $G$ by
(i) deleting all vertices of $P$ except $v$ (to enforce simplicity), and
(ii) deleting all graph edges $(\ell(h), \ell(h'))$ for each child $h'$ of $h$ in $T$ (to avoid reproducing already stored branches);
if a $v$--$t$ path exists in $H$, it inserts into $Q$ the path $D=P\oplus S$, where $S$ is a shortest $v$--$t$ path in $H$. Vertices labeled $t$ are leaves and are skipped, as they represent complete paths admitting no deviation.
 
We maintain the following invariant throughout the while-loop:
\begin{quote}
\emph{At the beginning of each iteration, $Q$ contains exactly one entry for each vertex $h$ of the current pseudo-tree $T$ with $\ell(h)\neq t$ that admits a feasible deviation, namely the candidate $D$ produced by $\expand(T,G,h,\cdot)$, and no other element. Consequently, the minimum-weight element of $Q$ has minimum weight among all simple $s$--$t$ paths not yet in \sol.}
\end{quote}
 
\noindent\textbf{Initialization.}
The seeding loop invokes \expand on each vertex of the initial pseudo-tree exactly once, inserting into $Q$ the shortest feasible deviation for the corresponding prefix, if any.
Consider any simple $s$--$t$ path $R$ not in \sol.
Let $h$ be the unique vertex of $T$ whose root-to-$h$ label sequence is the longest prefix of $R$ represented in $T$. Since $R$ is simple and not represented in $T$, this prefix ends before $t$, i.e. $\ell(h)=v\neq t$, so $h$ is not skipped by \expand.
By maximality of the prefix, the next edge of $R$ after $v$ does not correspond to any pseudo-tree edge $(h, h')$, and therefore it is not removed when constructing the subgraph $H$ for $h$.
Moreover, since $R$ is simple, it avoids all vertices of $P \setminus \{v\}$.
Thus $R$ is feasible in $H$, and \expand inserts a candidate $D=P\oplus S$ whose weight satisfies $\wg{D}\le \wg{R}$, because $S$ is a shortest $v$--$t$ path in $H$.
It follows that the minimum-weight element of $Q$ has weight no larger than any simple $s$--$t$ path not in \sol, and the invariant holds initially.
 
\noindent\textbf{Maintenance.}
At each iteration, the algorithm extracts from $Q$ a minimum-weight candidate path $\Pi$ and adds it to \sol.
By the invariant, $\Pi$ has minimum weight among all simple $s$--$t$ paths not yet in \sol, and is therefore the next simple shortest path.
The algorithm then merges $\Pi$ into the pseudo-tree, hanging below the deviation vertex $h$ recorded with $\Pi$ only the previously missing suffix, and returns the set $W$ of vertices to be re-expanded.
The merge changes the pseudo-tree in two ways only: the vertex below which the suffix is hung acquires one new child, and the newly created suffix vertices are added (each, except the leaf, with a single child along the new branch).
$W$ consists exactly of these vertices, together with $h$: the candidate that $h$ had contributed to $Q$ has just been extracted, so $h$ must be re-expanded to yield its next-best deviation.
The two coincide unless ties caused $T$ to grow below $h$ before the extraction, in which case \merge\ descends to the vertex that actually acquires the new child and both are included in $W$.
By invoking \expand on every vertex of $W$, the algorithm recomputes the shortest feasible deviation of the branching vertex, now that one further outgoing edge (toward the new child) is forbidden, and computes the first deviation of each newly created vertex. All vertices outside $W$ keep the same set of children, hence their candidates in $Q$ remain the shortest feasible ones and need not be recomputed.
Therefore the invariant is restored for the next iteration, and the minimum-weight element of $Q$ is again the next simple shortest path.
 
Finally, no complete $s$--$t$ path is ever added to \sol\ more than once.
Fix a state of the pseudo-tree. Then any simple path $R$ can be generated only from the unique vertex corresponding to the longest prefix of $R$ represented in $T$: from any proper ancestor, the next edge of $R$ corresponds to a pseudo-tree edge to one of its children and is therefore removed, while from any other vertex the generated prefix differs from that of $R$.
In particular, a path already represented in $T$, hence any path already in \sol, cannot be generated at all, and the deviation produced when the branching vertex is re-expanded after a merge avoids the newly added child, hence differs from $\Pi$.

Across different states of the pseudo-tree, however, the vertex realizing the longest represented prefix of $R$ changes: after a \merge, the newly created suffix vertices may realize a longer prefix of $R$ than the vertex that generated $R$ earlier, and re-expanding them can generate $R$ a second time.
This is possible only when the two candidates have equal weight.
Indeed, let $R$ be generated from $h$ and let $\Pi$ be the path whose merge creates a descendant of $h$ that can regenerate $R$; then $\Pi$ leaves $h$ through the same edge as $R$, so $\Pi$ was a feasible deviation of $h$ when $h$ was last expanded, whence $\wg{R} \leq \wg{\Pi}$ by minimality of \expand. If the inequality is strict, $R$ is extracted before $\Pi$ is merged and is represented in $T$ from that moment on, so it cannot be regenerated; if $\wg{R} = \wg{\Pi}$ and the queue extracts $\Pi$ first, a second copy of $R$ may be produced.
The membership set \found\ discards it, and no candidate is lost: the extraction of the copy already in $Q$ triggers a \merge\ whose branching vertex is exactly the vertex that would have regenerated $R$, which is thus re-expanded and contributes its next-best deviation.
Since $Q$ contains at most one copy of each path and every extracted path is immediately represented in $T$, duplicates are impossible.
 
\noindent\textbf{Termination.}
The loop terminates either when $|\sol|=k'$ or when $Q$ becomes empty.
In the first case, the algorithm has appended $k'-k$ successive simple shortest paths to the initial \sol, and therefore \sol is a valid top-$k'$ solution.
In the second case, by the invariant, no further simple $s$--$t$ path exists outside \sol, and the algorithm correctly returns all simple $s$--$t$ paths.
\end{proof}
\begin{theorem}[Time Complexity]
\label{thm:dynk-time}
Algorithm~\ref{algo:dynamic_k} runs in $\bigO(k + k'n \cdot \textsc{sp}(n,m))$ worst-case time; more precisely, it runs in $\bigO(k)$ time if $k'<k$, and in $\bigO\big((L + (k'-k)n) \cdot \textsc{sp}(n,m)\big)$ time if $k'>k$, where $L$ is the total length of the input collection \sol.
\end{theorem}
\begin{proof}
If $k'<k$, the algorithm returns the $k'$ lightest paths among the already available $k$ ones. No assumption on the order of the input collection is needed: using a worst-case linear-time selection algorithm (e.g., quickselect with median-of-medians pivoting~\cite{BlumFPRT73}), the $k'$ lightest paths can be identified in $\bigO(k)$ time, which is accounted for by the additive $\bigO(k)$ term of the claimed bound; this is also optimal for this case, as any correct algorithm must inspect the weights of all $k$ input paths.
 
Assume now $k'>k$; we prove the refined bound $\bigO\big((L + (k'-k)n)\cdot\textsc{sp}(n,m)\big)$, from which the unified bound follows since $L \leq kn$ implies $L + (k'-k)n \leq k'n$ (and reading the input costs $\bigO(L)$, absorbed below). Building the pseudo-tree from \sol takes $\bigO(n + L)$ worst-case time by Theorem~\ref{thm:pseudo-tree-wc-time} (using Algorithm~\ref{alg:pseudo-tree-wc}); as shown below, this cost is dominated by the seeding of the queue, since $\textsc{sp}(n,m) = \Omega(n)$.
Each call to \expand performs $\bigO(n)$ work to construct the subgraph $H$ (removing the at most $n$ prefix vertices and the graph edges corresponding to the pseudo-tree edges from $h$ to its children) and one shortest-path computation; hence it costs $\bigO(n + \textsc{sp}(n,m)) = \bigO(\textsc{sp}(n,m))$ time, since $\textsc{sp}(n,m) = \Omega(n)$. Note that the assignment $H\gets G$ in Procedure~\ref{proc:expand} is a notational convenience: $H$ is never materialized as an explicit copy of $G$, which would cost $\Theta(n+m)$ per call. In the implementation, prefix vertices and forbidden outgoing edges are simply flagged as temporarily removed in $G$ before the shortest-path query and unflagged right after it, at a cost proportional to the number of such elements, i.e., $\bigO(n)$ here.
The membership test against \found\ adds $\bigO(|D|) = \bigO(n)$ expected time per call, as \found\ is a hash set keyed on the vertex sequence of the candidate; this too is dominated by $\textsc{sp}(n,m)$, so the bound is unaffected.
The seeding loop invokes \expand once for each vertex of the initial pseudo-tree. As the latter has at most $L$ vertices (Theorem~\ref{thm:pseudo-tree-space}), seeding costs $\bigO(L \cdot \textsc{sp}(n,m))$, which dominates the $\bigO(n+L)$ construction cost.
In the while-loop, each of the $k'-k$ paths added beyond the initial $k$ triggers one call to \merge, costing $\bigO(|S|) = \bigO(n)$ (the appended suffix has at most $n$ vertices and no search for the longest represented prefix is performed), which returns a set $W$ of at most $n+2$ vertices (the deviation vertex $h$, the vertex acquiring the new child, and at most $n$ newly created vertices). The algorithm then invokes \expand on each of them. Hence the while-loop performs $\bigO((k'-k)n)$ calls to \expand, costing $\bigO((k'-k)n \cdot \textsc{sp}(n,m))$.
Summing all contributions, the total running time for $k'>k$ is $\bigO\big(n + L + L \cdot \textsc{sp}(n,m) + (k'-k)n \cdot \textsc{sp}(n,m)\big) = \bigO\big((L + (k'-k)n) \cdot \textsc{sp}(n,m)\big)$, using $\textsc{sp}(n,m) = \Omega(n)$ to absorb the additive $\bigO(n+L)$ term. Together with the $\bigO(k)$ cost of the case $k'<k$, and since $L \leq kn$, the unified $\bigO(k + k'n\cdot\textsc{sp}(n,m))$ bound follows.
\end{proof}
Note that the dominant term in the running time for $k'>k$ comes from the seeding phase, which re-derives the shortest feasible deviation of every one of the $\bigO(L)$ pseudo-tree vertices. 
The bound is thus comparable to recomputation in the worst case, and improves on it whenever $L \ll kn$, as is typical in practice (Section~\ref{sec:experiments}).
\begin{remark}\label{rem:any-collection}
\dynk relies neither on \sol being a top-$k$ solution nor on how it was computed: \expand generates, at each pseudo-tree vertex, the lightest simple $s$--$t$ path leaving $T$ there, so the queue minimum is always the lightest simple $s$--$t$ path not in \sol. Hence, started from \emph{any} collection of simple $s$--$t$ paths and run until the queue minimum reaches the weight of the $k'$-th lightest path in \sol (the stopping rule of \capheu below), it yields a top-$k'$ solution once the paths heavier than the $k'$-th are purged; in particular, it can resume a solver-independent partial enumeration such as the one produced by the single-source algorithm of~\cite{DEmidioS26}.
\end{remark}
Note also that \dynk admits several variants that share this skeleton but differ in how detours are computed, and that combine with standard \spksispshort preprocessing; we evaluate them in Sec.~\ref{sec:experiments}. The only point at which \dynk inspects the graph is the detour computation inside \expand. As presented, \expand computes it \emph{from scratch}, via a single-pair shortest-path call on $H$ as in \yenalgo; we call this instantiation \dynkyen. The detour computation, however, is independent of the rest of the algorithm and can be replaced by any subroutine that returns, for a given prefix and a given set of forbidden continuations at $v$, a shortest simple $v$--$t$ detour: the pseudo-tree representation, the queue-driven update loop, and the order in which paths are produced are unaffected by this choice.
Thus, alternatively, the detour subroutine can follow \pnc\ (Section~\ref{sec:prelim}): its shortest-path tree towards $t$ and residual weights are precomputed once and reused across all expansions, with lower-bound-guided generation and postponed simplicity repair, instead of a from-scratch shortest-path call per pseudo-tree vertex. This yields the variant \dynkpnc, which is correct as well: for each prefix it returns the same shortest feasible detour as \dynkyen, hence the same paths in the same non-decreasing weight order. Generally, \dynk can be seen as a \emph{scheme} parameterized by the detour subroutine, and different subroutines trade preprocessing cost against per-detour cost.
Orthogonally, note that every variant can be sped up by first applying the \peek\ graph-reduction preprocessing (Section~\ref{sec:prelim}): the chosen \dynk\ variant is then run on the reduced graph instead of on $G$. Since the reduction discards only vertices and edges that cannot lie on any top-$k'$ simple $s$--$t$ path, the output is preserved; as with any such algorithm-agnostic preprocessing, this affects only efficiency, not correctness.
Finally, we formalize a lightweight \emph{online} pruning rule which, to the best of our knowledge, has not been described in the literature, and which can be layered on top of every algorithm considered so far; we call it \capheu, for \emph{Cheapest-Alternatives Pruning}. Let $\kappa$ denote the number of paths sought ($\kappa = k$ for a from-scratch execution, $\kappa = k'$ for an update). Alongside the priority queue of candidates, the algorithm maintains a max-heap $M$ of capacity $\kappa$: whenever a candidate path is generated, its weight is also inserted into $M$ and, upon the $(\kappa+1)$-th insertion, the maximum of $M$ is discarded, so that $M$ retains, at all times, the weights of the $\kappa$ \emph{lightest} pairwise distinct candidate paths generated so far. Once $M$ reaches size $\kappa$, its maximum $B$ (the $\kappa$-th lightest such weight) upper-bounds the weight of every path of \emph{some} top-$\kappa$ solution: the paths accounted for in $M$ are $\kappa$ pairwise distinct valid simple $s$--$t$ paths of weight at most $B$, so, by Def.~\ref{def:spksisp}, a valid solution exists that excludes any path of weight $\geq B$ (ties are broken in favor of the paths of $M$). This enables three prunings: (i) any newly generated candidate of weight $\geq B$ is discarded instead of being enqueued; (ii) any detour computation is aborted as soon as the weight of the prefix plus the partial (monotonically non-decreasing) distance reaches $B$; this applies both to the shortest-path calls of Yen-style detours and to the shortest-path-tree-guided generation of \pnc-style ones; (iii) since $B$ is non-increasing over time, pruning becomes progressively more aggressive as the execution proceeds. As with \peek, this affects only efficiency, never correctness or the produced order.
\capheu composes naturally with the rest of our toolbox. Combined with \peek, $M$ can be pre-filled with the (up to $\kappa$) paths produced while computing \peek's upper bound, activating the pruning from the very start. Combined with \dynk, $M$ is pre-filled with the paths of the input top-$k$ solution and maintained with capacity $\kappa = k'$: for $k'>k$ the bound activates as soon as $k'$ distinct candidates are known, and every seeding-phase expansion already benefits from prunings (i) and (ii).
In Section~\ref{sec:experiments} we compare variants \dynkyen and \dynkpnc, with and without the \peek preprocessing and the \capheu pruning, to identify the configurations that perform best in practice; the evaluation also covers the pseudo-tree construction algorithms.
\section{Experimental Evaluation}
\label{sec:experiments}
To evaluate the algorithms from Section~\ref{sec:algorithms}, we implemented them and conducted an extensive experimental study with two main objectives. First, we compare the pseudo-tree construction algorithms. In fact, although \buildwc runs in $O(n+L)$ \emph{worst-case} time, unlike the $O(L)$ \emph{expected} bound of \build, the inputs underlying its analysis appear uncommon in practice, since typical \spksispshort solutions tend to contain much shorter, highly prefix-sharing paths. Thus, we evaluate which construction performs better on realistic instances.
We then assess the practical benefit of \dynk. Although its worst-case time is comparable to recomputation (Theorem~\ref{thm:dynk-time}), in practice only a small subset of pseudo-tree vertices induce feasible deviations, making updates cheaper; we quantify this advantage, compare \dynkyen and \dynkpnc, and evaluate the effect of \capheu and \peek-based preprocessing.
\blocco{Implementation and Setup}
\begin{table}[t]
\centering
\renewcommand{\arraystretch}{0.8}
\footnotesize
\begin{tabular}{c|c|c|cc|ccc|cc|c}
\textbf{Dataset} & \textbf{Graph} & \textbf{Type} & $|V|$ & $|E|$ & \textbf{D} & \textbf{W} &\textbf{S}& \avgdeg & \maxdeg & \diam \\
\midrule
\caida & \caidass & \textsc{internet} & 32\,000 & 40\,204 & \false{2pt} & \true{2pt} & \false{2pt} & 2.51 & 203 & 31 \\
\lux & \luxss & \textsc{road} & 30\,647 & 75\,546 & \true{2pt} & \true{2pt} & \false{2pt} & 4.93 & 18 & 253 \\
\efz & \efzss & \textsc{railway} & 117\,942 & 380\,194 & \true{2pt} & \true{2pt} & \false{2pt} & 6.45 & 44 & 163 \\
\epini & \epiniss & \textsc{social} & 41\,441 & 693\,507 & \true{2pt} & \false{2pt} & \false{2pt} & 33.47 & 3622 & 11 \\
\dnk & \dnkss & \textsc{road} & 469\,110 & 1\,090\,038 & \true{2pt} & \true{2pt} & \false{2pt} & 4.65 & 14 & 640 \\
\forest & \forestss & \textsc{forest-fire} & 2\,000\,000 & 14\,908\,267 & \false{2pt} & \false{2pt} & \true{2pt} & 14.91 & 41\,879 & 24 \\
\end{tabular}
\caption{\small Overview of the input graphs: dataset name, acronym, and network type; number of vertices and edges; flags \textbf{D}/\textbf{W}/\textbf{S} for directed/weighted/synthetic (\true{2pt}/\false{2pt}); average degree (\avgdeg${}=2|E|/|V|$), maximum degree (\maxdeg), and diameter (\diam, estimated for the largest instances). For disconnected graphs, values refer to the largest (strongly) connected component, on which experiments run. Sorted by non-decreasing $|E|$.}
\label{table:input}
\end{table}
\begin{table}[t]
\centering
\caption{\small Running time (ms) of the two construction algorithms vs.\ the number $k$ of input paths; the fifth column is the slowdown of \buildwc over \build, the last the average pseudo-tree size. Values are averaged over all runs with the given $k$ in the reference configuration.}
\label{table:pseudotree-times}
\renewcommand{\arraystretch}{0.5}
\setlength{\tabcolsep}{4.5pt}
\footnotesize
\begin{tabular}{c|c|rr|c|r}
\textbf{Graph} & $k$ & \build & \buildwc & \shortstack{\buildwc\,/\\ \build} & $|V_T|$ \\
\midrule
\multirow{5}{*}{\caidass} & $50$ & $0.06$ & $1.02$ & $16.0\times$ & $422$ \\
 & $100$ & $0.12$ & $1.17$ & $9.8\times$ & $799$ \\
 & $200$ & $0.23$ & $1.39$ & $6.0\times$ & $1\,545$ \\
 & $400$ & $0.48$ & $1.76$ & $3.7\times$ & $3\,038$ \\
 & $800$ & $0.99$ & $2.57$ & $2.6\times$ & $6\,009$ \\
\midrule
\multirow{5}{*}{\luxss} & $50$ & $0.34$ & $1.62$ & $4.8\times$ & $2\,660$ \\
 & $100$ & $0.60$ & $2.14$ & $3.6\times$ & $4\,881$ \\
 & $200$ & $1.12$ & $3.22$ & $2.9\times$ & $8\,758$ \\
 & $400$ & $2.20$ & $4.96$ & $2.3\times$ & $16\,396$ \\
 & $800$ & $4.61$ & $8.37$ & $1.8\times$ & $31\,102$ \\
\midrule
\multirow{5}{*}{\efzss} & $50$ & $0.17$ & $3.95$ & $22.7\times$ & $1\,221$ \\
 & $100$ & $0.37$ & $4.25$ & $11.6\times$ & $2\,269$ \\
 & $200$ & $0.71$ & $4.99$ & $7.0\times$ & $4\,301$ \\
 & $400$ & $1.39$ & $6.16$ & $4.4\times$ & $8\,356$ \\
 & $800$ & $2.82$ & $8.64$ & $3.1\times$ & $16\,302$ \\
\midrule
\multirow{5}{*}{\epiniss} & $50$ & $0.02$ & $1.12$ & $54.2\times$ & $164$ \\
 & $100$ & $0.04$ & $1.20$ & $28.3\times$ & $319$ \\
 & $200$ & $0.10$ & $1.37$ & $13.6\times$ & $646$ \\
 & $400$ & $0.20$ & $1.70$ & $8.5\times$ & $1\,297$ \\
 & $800$ & $0.50$ & $2.16$ & $4.3\times$ & $2\,603$ \\
\midrule
\multirow{5}{*}{\dnkss} & $50$ & $0.90$ & $16.89$ & $18.7\times$ & $10\,022$ \\
 & $100$ & $1.60$ & $18.20$ & $11.4\times$ & $18\,148$ \\
 & $200$ & $2.90$ & $21.48$ & $7.4\times$ & $33\,481$ \\
 & $400$ & $5.44$ & $28.22$ & $5.2\times$ & $61\,822$ \\
 & $800$ & $10.34$ & $39.84$ & $3.9\times$ & $114\,864$ \\
\midrule
\multirow{5}{*}{\forestss} & $50$ & $0.06$ & $90.00$ & $1609.0\times$ & $180$ \\
 & $100$ & $0.08$ & $91.36$ & $1141.0\times$ & $346$ \\
 & $200$ & $0.15$ & $91.06$ & $601.2\times$ & $689$ \\
 & $400$ & $0.31$ & $94.55$ & $304.5\times$ & $1\,371$ \\
 & $800$ & $0.65$ & $95.78$ & $147.5\times$ & $2\,734$ \\
\end{tabular}
\end{table}
We implemented all algorithms in Rust, using \texttt{petgraph} for graph operations; the pseudo-tree, its constructions (\build, \buildwc), and all dynamic variants (\dynkyen, \dynkpnc, with and without \peek) were written from scratch.
Our implementation of \yenalgo incorporates Lawler's classical modification~\cite{Lawler72}, which restricts deviation vertices of a newly generated candidate to those lying at or after the vertex at which the path being extended differs from the original, thus avoiding the regeneration of already enumerated candidates. The code was compiled with Cargo~1.93.0 in optimized  mode.
Experiments were run on a QEMU/KVM virtual machine with $32$ vCPUs and $192$\,GiB RAM, hosted on a dual-socket server with two AMD EPYC~9554 processors ($64$ cores each, $3.10$\,GHz) and $1.5$\,TiB RAM.

\blocco{Datasets}
We considered a heterogeneous collection of input graphs, whose characteristics are summarized in Table~\ref{table:input}. The core of our benchmark consists of real-world networks spanning diverse topologies and application domains taken from public repositories~\cite{nr,demetrescu2009shortest}. We focus on domains where top-$k$ simple shortest paths are a relevant primitive: road, railway, internet, and communication networks~\cite{CDD24,AkibaHNIY15,cionini_et_al:OASIcs.ATMOS.2014.46}.
The benchmark also includes unweighted instances (a directed social network and a large synthetic graph generated with the Forest Fire model~\cite{LeskovecKF07}), so that it spans both the weighted/unweighted and the directed/undirected dichotomies, which our algorithms all support.
\blocco{Experimental Protocol}
For each graph in Table~\ref{table:input}, we selected $10$ source--target pairs uniformly at random and report average results over such pairs.
For the initial solution size we considered $k \in \{50,100,200,400,800\}$, increasing by repeated doubling as customary in Algorithm Engineering~\cite{AngrimanGLMNPT19}; each top-$k$ solution was then updated to $k' \in \{3k/2,\,2k,\,4k\}$ and compared against recomputing the top-$k'$ solution from scratch.
Every execution in our pipeline is subject to a $30$-minute time limit, to keep running times manageable on the largest inputs. For each pair, we first compute a top-$k$ solution with \yenalgo, falling back to \pnc\ under the same limit if \yenalgo\ times out (had both timed out the pair would have been resampled, which never occurred). The solution is then fed separately to \dynkyen and \dynkpnc, under the same limit; each update is run only if its matching from-scratch algorithm completed in time (so every dynamic variant has its baseline), and configurations where it did not are marked unavailable.
The comparison between the two pseudo-tree construction algorithms follows the same protocol: for each sampled pair and each value of $k$, we measure the time taken by \build and \buildwc to build the pseudo-tree of the initial top-$k$ solution.
\begin{table}[t]
\centering
\renewcommand{\arraystretch}{0.8}
\setlength{\tabcolsep}{3.5pt}
\caption{\small Cumulative running times (seconds over the $10$ queried pairs) of from-scratch recomputation (\yenalgo, \pnc) and dynamic update (\dynkyen, \dynkpnc) to top-$k'$, in the reference configuration (\capheu on, \peek off, \build). For \dnkss, where \yenalgo\ times out for all $k$, only \pnc-based columns appear. Bold: fastest per graph; ``--'': unavailable (cf.\ the protocol); $\dagger$: cumulated over the pairs where the from-scratch algorithm did not time out, keeping its from-scratch and update times comparable.}
\label{table:dyn-all}
\resizebox{\textwidth}{!}{%
\begin{tabular}{c|c|rrrr|rrrr|rrrr|rrrr|rr|rrrr}
 & & \multicolumn{4}{c|}{\caidass} & \multicolumn{4}{c|}{\luxss} & \multicolumn{4}{c|}{\efzss} & \multicolumn{4}{c|}{\epiniss} & \multicolumn{2}{c|}{\dnkss} & \multicolumn{4}{c}{\forestss} \\
$k$ & $k'$ & \multicolumn{1}{c}{\yenalgo} & \multicolumn{1}{c}{\dynkyen} & \multicolumn{1}{c}{\pnc} & \multicolumn{1}{c|}{\dynkpnc} & \multicolumn{1}{c}{\yenalgo} & \multicolumn{1}{c}{\dynkyen} & \multicolumn{1}{c}{\pnc} & \multicolumn{1}{c|}{\dynkpnc} & \multicolumn{1}{c}{\yenalgo} & \multicolumn{1}{c}{\dynkyen} & \multicolumn{1}{c}{\pnc} & \multicolumn{1}{c|}{\dynkpnc} & \multicolumn{1}{c}{\yenalgo} & \multicolumn{1}{c}{\dynkyen} & \multicolumn{1}{c}{\pnc} & \multicolumn{1}{c|}{\dynkpnc} & \multicolumn{1}{c}{\pnc} & \multicolumn{1}{c|}{\dynkpnc} & \multicolumn{1}{c}{\yenalgo} & \multicolumn{1}{c}{\dynkyen} & \multicolumn{1}{c}{\pnc} & \multicolumn{1}{c}{\dynkpnc} \\
\midrule
\multirow{3}{*}{$50$} & $75$ & $1.34$ & $1.11$ & $0.19$ & $\mathbf{0.16}$ & $5.94$ & $5.49$ & $0.23$ & $\mathbf{0.17}$ & $9.16$ & $7.80$ & $2.10$ & $\mathbf{1.95}$ & $0.03$ & $\mathbf{0.03}$ & $0.91$ & $0.95$ & $3.08$ & $\mathbf{2.45}$ & $1.32$ & $\mathbf{1.07}$ & $57.60$ & $54.43$ \\
 & $100$ & $1.67$ & $1.48$ & $0.32$ & $\mathbf{0.27}$ & $6.86$ & $6.52$ & $0.29$ & $\mathbf{0.20}$ & $11.14$ & $10.21$ & $3.55$ & $\mathbf{3.39}$ & $0.05$ & $\mathbf{0.04}$ & $0.87$ & $0.81$ & $3.32$ & $\mathbf{2.65}$ & $2.08$ & $\mathbf{1.65}$ & $52.81$ & $55.14$ \\
 & $200$ & $2.99$ & $2.78$ & $0.50$ & $\mathbf{0.48}$ & $10.13$ & $9.66$ & $0.58$ & $\mathbf{0.36}$ & $19.90$ & $17.76$ & $4.06$ & $\mathbf{3.81}$ & $0.10$ & $\mathbf{0.09}$ & $0.87$ & $0.78$ & $5.21$ & $\mathbf{3.06}$ & $\mathbf{2.68}$ & $3.12$ & $54.59$ & $59.36$ \\
\midrule
\multirow{3}{*}{$100$} & $150$ & $2.22$ & $1.82$ & $0.39$ & $\mathbf{0.31}$ & $8.68$ & $8.07$ & $0.41$ & $\mathbf{0.25}$ & $14.42$ & $12.36$ & $3.93$ & $\mathbf{3.55}$ & $0.08$ & $\mathbf{0.06}$ & $1.03$ & $1.13$ & $4.52$ & $\mathbf{2.74}$ & $2.47$ & $\mathbf{1.52}$ & $55.43$ & $55.90$ \\
 & $200$ & $3.06$ & $2.67$ & $0.52$ & $\mathbf{0.45}$ & $10.05$ & $9.53$ & $0.58$ & $\mathbf{0.34}$ & $19.58$ & $16.98$ & $4.03$ & $\mathbf{3.70}$ & $0.10$ & $\mathbf{0.08}$ & $0.97$ & $0.95$ & $5.34$ & $\mathbf{3.00}$ & $2.62$ & $\mathbf{2.17}$ & $56.34$ & $57.47$ \\
 & $400$ & $6.10$ & $5.56$ & $0.96$ & $\mathbf{0.87}$ & $16.75$ & $15.52$ & $1.37$ & $\mathbf{0.62}$ & $34.20$ & $31.18$ & $6.25$ & $\mathbf{5.60}$ & $0.20$ & $\mathbf{0.18}$ & $1.08$ & $1.01$ & $12.00$ & $\mathbf{3.11}$ & $6.01$ & $\mathbf{4.71}$ & $57.32$ & $54.78$ \\
\midrule
\multirow{3}{*}{$200$} & $300$ & $4.65$ & $3.91$ & $0.83$ & $\mathbf{0.51}$ & $13.27$ & $12.39$ & $0.96$ & $\mathbf{0.43}$ & $26.89$ & $21.92$ & $5.10$ & $\mathbf{4.24}$ & $0.16$ & $\mathbf{0.10}$ & $1.00$ & $0.99$ & $8.15$ & $\mathbf{3.12}$ & $3.91$ & $\mathbf{2.73}$ & $55.25$ & $54.63$ \\
 & $400$ & $6.00$ & $5.18$ & $0.93$ & $\mathbf{0.78}$ & $16.71$ & $15.42$ & $1.37$ & $\mathbf{0.56}$ & $34.08$ & $28.32$ & $6.09$ & $\mathbf{5.25}$ & $0.21$ & $\mathbf{0.16}$ & $1.07$ & $1.07$ & $11.77$ & $\mathbf{3.39}$ & $6.24$ & $\mathbf{5.09}$ & $59.26$ & $55.30$ \\
 & $800$ & $11.57$ & $10.55$ & $1.70$ & $\mathbf{1.43}$ & $28.46$ & $25.33$ & $4.08$ & $\mathbf{1.15}$ & $58.55$ & $54.22$ & $12.34$ & $\mathbf{10.44}$ & $0.44$ & $\mathbf{0.37}$ & $1.15$ & $1.11$ & $35.23$ & $\mathbf{3.95}$ & $14.22$ & $\mathbf{12.84}$ & $60.35$ & $58.64$ \\
\midrule
\multirow{3}{*}{$400$} & $600$ & $8.68$ & $7.10$ & $1.37$ & $\mathbf{0.84}$ & $22.53$ & $19.40$ & $2.57$ & $\mathbf{0.75}$ & $45.58$ & $35.93$ & $10.49$ & $\mathbf{7.43}$ & $0.30$ & $\mathbf{0.19}$ & $1.15$ & $1.04$ & $21.95$ & $\mathbf{3.75}$ & $10.12$ & $\mathbf{6.67}$ & $60.95$ & $59.05$ \\
 & $800$ & $11.60$ & $9.82$ & $1.72$ & $\mathbf{1.14}$ & $28.93$ & $25.05$ & $4.13$ & $\mathbf{1.01}$ & $59.34$ & $46.89$ & $12.35$ & $\mathbf{8.28}$ & $0.43$ & $\mathbf{0.30}$ & $1.20$ & $1.01$ & $35.16$ & $\mathbf{3.67}$ & $14.07$ & $\mathbf{10.89}$ & $58.45$ & $60.67$ \\
 & $1600$ & $22.59$ & $19.99$ & $3.68$ & $\mathbf{2.69}$ & $52.13$ & $40.84$ & $12.39$ & $\mathbf{2.08}$ & $111.34$ & $101.49$ & $30.05$ & $\mathbf{22.52}$ & $1.09$ & $\mathbf{0.73}$ & $1.56$ & $1.27$ & $133.68$ & $\mathbf{5.42}$ & $26.03$ & $\mathbf{22.22}$ & $65.47$ & $66.19$ \\
\midrule
\multirow{3}{*}{$800$} & $1200$ & $17.67$ & $13.85$ & $2.59$ & $\mathbf{1.44}$ & $40.63$ & $32.91$ & $7.74$ & $\mathbf{1.39}$ & $88.96$ & $63.98$ & $19.03$ & $\mathbf{9.78}$ & $0.72$ & $\mathbf{0.38}$ & $1.43$ & $1.21$ & $74.83$ & $\mathbf{3.91}$ & $21.11$ & $\mathbf{13.63}$ & $62.64$ & $56.31$ \\
 & $1600$ & $22.88$ & $18.27$ & $3.72$ & $\mathbf{2.23}$ & $52.39$ & $39.93$ & $12.34$ & $\mathbf{1.84}$ & $109.38$ & $81.46$ & $29.78$ & $\mathbf{16.75}$ & $1.03$ & $\mathbf{0.66}$ & $1.55$ & $1.28$ & $133.99$ & $\mathbf{5.41}$ & $28.71$ & $\mathbf{20.42}$ & $64.33$ & $61.46$ \\
 & $3200$ & $48.46$ & $41.05$ & $8.84$ & $\mathbf{4.94}$ & $114.90$ & $66.85$ & $51.33$ & $\mathbf{3.79}$ & $211.43$ & $180.11$ & $66.42$ & $\mathbf{40.62}$ & $2.44$ & $\mathbf{1.51}$ & $2.36$ & $1.69$ & $574.91$ & $\mathbf{8.40}$ & $66.29$ & $\mathbf{49.77}$ & $86.19$ & $65.19$ \\
\end{tabular}
}
\end{table}

\begin{table}[t]
\centering
\renewcommand{\arraystretch}{0.9}
\setlength{\tabcolsep}{2.6pt}
\caption{\small Effect of the heuristics: geometric-mean factors over all $(k,k')$ by which each combination, i.e., \capheu only (\textsc{c}), \peek only (\textsc{p}), both (\textsc{c+p}), scales the mean running time of the heuristic-free base; values above $1$ denote speedup. $\dagger$: \dnkss\ \yenalgo-based factors are over $k\in\{50,100\}$ due to timeouts.}
\label{table:ablation}
\footnotesize
\begin{tabular}{c|ccc|ccc|ccc|ccc}
 & \multicolumn{3}{c|}{\yenalgo} & \multicolumn{3}{c|}{\dynkyen} & \multicolumn{3}{c|}{\pnc} & \multicolumn{3}{c}{\dynkpnc} \\
\textbf{Graph} & \textsc{c} & \textsc{p} & \textsc{c+p} & \textsc{c} & \textsc{p} & \textsc{c+p} & \textsc{c} & \textsc{p} & \textsc{c+p} & \textsc{c} & \textsc{p} & \textsc{c+p} \\
\midrule
\caidass & $1.4$ & $0.9$ & $1.2$ & $1.5$ & $0.9$ & $1.3$ & $1.2$ & $0.8$ & $1.0$ & $1.3$ & $0.8$ & $1.0$ \\
\luxss & $1.4$ & $1.0$ & $1.2$ & $1.4$ & $1.0$ & $1.2$ & $1.3$ & $0.9$ & $1.0$ & $1.6$ & $0.8$ & $1.1$ \\
\efzss & $3.4$ & $2.8$ & $3.4$ & $3.9$ & $2.9$ & $3.5$ & $1.8$ & $1.3$ & $1.4$ & $2.2$ & $1.4$ & $1.5$ \\
\epiniss & $1.0$ & $0.1$ & $0.1$ & $1.0$ & $0.1$ & $0.1$ & $1.0$ & $0.7$ & $0.7$ & $1.0$ & $0.6$ & $0.6$ \\
\dnkss$^\dagger$ & $1.1$ & $1.3$ & $1.3$ & $1.0$ & $1.3$ & $1.4$ & $1.2$ & $0.8$ & $0.9$ & $1.8$ & $0.6$ & $0.8$ \\
\forestss & $0.9$ & $0.0$ & $0.0$ & $0.9$ & $0.0$ & $0.0$ & $1.0$ & $0.4$ & $0.4$ & $1.0$ & $0.4$ & $0.4$ \\
\end{tabular}
\end{table}

For every $(k,k')$ we ran, on \build, the full $2\times 2$ matrix of \peek and \capheu on/off; \buildwc was additionally run at (\peek off, \capheu on), matching the \build inputs, to compare construction times. Unless stated otherwise, (\peek off, \capheu on) with \build is the \emph{reference configuration}, to which reported pseudo-tree sizes refer.

For each update we report the \emph{total} running time, comprising all phases, including the pseudo-tree construction and the seeding loop over its vertices (Algorithm~\ref{algo:dynamic_k}), and also measure these two components separately, alongside the from-scratch times of \yenalgo\ and \pnc. Each dynamic variant is compared against its matching from-scratch configuration (same detour subroutine and preprocessing, e.g.\ \dynkpnc vs.\ \pnc). We further collect machine-independent measures: the pseudo-tree size, the number of \expand invocations, and the number of single-pair shortest-path computations, which for the Yen-based variants equals the \expand count, whereas the \pnc-based ones use the precomputed shortest-path tree towards $t$ to avoid part of them.
We deliberately \emph{exclude} one baseline: the trivial strategy that assumes the whole internal state of the from-scratch algorithm to be kept alongside the solution, making the update immediate: one resumes \yenalgo\ (or \pnc) and extracts $k'-k$ further paths, the counterpart of the observation following Theorem~\ref{thm:dynk-time}. This presumes persisting, per pair, \emph{both} the priority queue of candidate deviations and a membership set to discard duplicates, then reloading them in full at update time. Their footprint far exceeds the solution's: the candidates generated up to the $k$-th path are well approximated by the pseudo-tree size, which in our data is $6\times$ to over $50\times$ the value of $k$ (over $31\,000$ for $k=800$ on \luxss). Such state is moreover solver-specific, whereas the pseudo-tree is solver-agnostic and derivable from the solution alone. Our algorithms thus target the practically relevant setting where only the solution is retained across queries, and our experiments quantify exactly the cost of not storing solver state, i.e., the seeding phase.
\blocco{Results}
Table~\ref{table:pseudotree-times} compares the construction times of \build and \buildwc. On the weighted graphs, \build\ is faster in every instance, with a gap ranging from a few-fold (on \luxss) to three orders of magnitude on the large dense \forestss, where \buildwc's $\Theta(n)$ initialization over two million vertices dwarfs the tiny solution-dependent work of \build, and the ratio shrinking as $k$ grows. This is fully consistent with the refined bounds of Section~\ref{sec:algorithms}: \build\ runs in $O(L)$ expected time, visiting only vertices on the input paths, whereas \buildwc\ additionally pays the $\Theta(n)$ term of its $O(n+L)$ bound to initialize its global arrays, a dominant overhead on these instances, where $L \ll n$. Hence, harmless as it is asymptotically, that term makes \build\ consistently the preferable construction in practice.
Table~\ref{table:dyn-all} compares from-scratch and dynamic computation for all $(k,k')$ pairs in the reference configuration, reporting cumulative times over the $10$ queried pairs. In $89$ of the $90$ available settings, one of the dynamic variants is the fastest method (highlighted in bold), confirming that updating an existing solution consistently outperforms recomputation; the only exception (\forestss, $k=50$, $k'=200$) combines the smallest $k$ with a large jump to $k'$, so that the fixed cost of building and seeding the pseudo-tree is not yet amortized, and even there the update trails recomputation only marginally.
On the weighted graphs the \pnc-based variants are the fastest, mirroring the behavior of the underlying static algorithms~\cite{DEmidioS26}, and the speedups of \dynkpnc\ over \pnc\ grow markedly with both $k$ and $k'$. They are most dramatic on the large road network \dnkss, where recomputing a top-$k'$ solution from scratch is up to $68\times$ slower than updating it (for $k=800$, $k'=3200$: $575$s versus $8$s), and reach $13.5\times$ on \luxss, while remaining more moderate on \caidass\ and \efzss\ (up to about $1.9\times$). On \dnkss, \yenalgo\ exceeds the $30$-min limit for all tested $k$, so only the \pnc-based methods remain viable, exactly the regime our dynamic scheme is meant to serve; for that graph we report the \pnc\ columns only. The unweighted \epiniss\ and \forestss\ reverse the usual ranking: with unit weights the shortest-path-tree guidance of \pnc\ does not pay off, so \yenalgo\ is markedly faster than \pnc\ there, and \dynkyen\ is the method of choice; the dynamic update still improves on recomputation.
Table~\ref{table:ablation} disentangles the individual effect of the two heuristics across all tested $(k,k')$. The \capheu pruning alone accelerates the search on the weighted graphs (most markedly on \efzss, with mean gains up to $3.4\times$ on \yenalgo\ and $3.9\times$ on \dynkyen) at the negligible cost of a single bounded heap, and never harms it appreciably there. The \peek preprocessing alone is far less consistent: it helps in some regimes but slows the computation down in others (factors below $1$), as the cost of its bound (two shortest-path trees) is not always repaid; combining it with \capheu is rarely meaningfully better than \capheu alone. On the unweighted \forestss\ the abundance of equal-weight paths makes the $k$-th-lightest bound loose, so \capheu's effect is mixed; on the large road graph \dnkss\ it yields consistent, if moderate, gains, while \peek, harmful to the \pnc-based methods, speeds up the \yenalgo-based ones by about $1.3\times$.
Overall, the ablation recommends \capheu on and \peek off, our reference configuration, as the best or near-best choice in the vast majority of settings; within it, the faster detour routine is graph-dependent (\dynkpnc on the weighted networks, \dynkyen on the dense unweighted \forestss).
\section{Conclusions}
\label{sec:conclusions}
We studied \spksispshort in a dynamic setting, focusing on updates of parameter~$k$. Building on the pseudo-tree representation of sets of simple $s$--$t$ paths, whose construction from an unstructured solution we cast as trie building, we introduced \dynk, a parametric scheme that updates a top-$k$ solution to a top-$k'$ one without recomputing from scratch, and we established its correctness and complexity. Experiments show that its updates are the fastest option in the large majority of settings, and often substantially so, and that the online pruning rule \capheu accelerates from-scratch and dynamic algorithms alike; \dynk can moreover replace the from-scratch fallback in the single-source algorithm of~\cite{DEmidioS26}.
Future work includes extending the framework to dynamic changes of the graph topology (a direction made natural by the pseudo-tree characterization, which compactly localizes the paths affected by a topology update) and reducing update costs by maintaining auxiliary structures across queries. A further promising direction is to speed up the update itself in the few regimes where it does not already dominate: the marginal near-ties we observed suggest room for tighter pruning bounds and smarter branching rules that curb redundant detour generation. Finally, on the experimental side, we plan to extend the evaluation to larger graphs and to a wider range of values of $k$ and $k'$.

\section*{Acknowledgment}
The authors used Claude (Anthropic) model Opus 4.8 to assist in revising the language and presentation of this paper; all technical content was devised and verified by the authors.

\bibliographystyle{abbrvnat}
\bibliography{fullversion}

\end{document}